\documentclass[pra,aps,twocolumn,superscriptaddress]{revtex4}
\usepackage{tabularx, float}
\usepackage{enumitem,hyperref,graphicx,epic,eepic,epsfig,amsmath,amsfonts,amsthm,latexsym,amssymb,revsymb,verbatim,subfigure,color,yfonts,ifthen,psfrag,pst-grad,amsmath,pstricks}
\newcommand{\cF}{\mathcal{F}}
\newcommand{\cN}{\mathcal{N}}
\newcommand{\bop}{\mathcal{B}}

\newcommand{\Complex}{\mathbb{C}}
\newcommand{\id}{\mathsf{id}}

\usepackage{bm}
\usepackage{bbm}

\newcommand{\be}{\begin{eqnarray} \begin{aligned}}
\newcommand{\ee}{\end{aligned} \end{eqnarray} }
\newcommand{\benn}{\begin{eqnarray*} \begin{aligned}}
\newcommand{\eenn}{\end{aligned} \end{eqnarray*} }

\newcommand{\Syn}{\mathrm{Syn}}
\newcommand{\Ext}{\mbox{Ext}}

\newcommand{\bc}{\begin{center}}
\newcommand{\ec}{\end{center}}
\newcommand{\half}{\frac{1}{2}}

\newcommand{\tr}{\mathop{\mathrm{tr}}\nolimits}

\newtheorem{theorem}{Theorem}[section]

\newtheorem{lemma}[theorem]{Lemma}

\newtheorem{corollary}[theorem]{Corollary}

\newcommand*{\cA}{\mathcal{A}} 
\newcommand*{\cB}{\mathcal{B}}

\newcommand*{\cH}{\mathcal{H}}
\newcommand*{\cI}{\mathcal{I}}

\newcommand*{\cK}{\mathcal{K}}

\newcommand*{\cQ}{\mathcal{Q}}
\newcommand*{\cR}{\mathcal{R}}

\newcommand{\SH}{{\rm H}}
\newcommand{\hmin}{{\rm H_{\rm min}}}

\newcommand{\hmineps}{{\rm H}^\eps_{\rm min}}

\usepackage{amsfonts}

\def\Complex{\mathbb{C}}

\def\id{\mathbb{I}}

\def\01{\{0,1\}}
\newcommand*{\sbin}{\{0,1\}}

\newcommand{\eps}{\epsilon}
\newcommand{\ket}[1]{|#1\rangle}

\newcommand{\suppfig}{Supplementary Figures}
\newcommand{\supptab}{Supplementary Tables}
\newcommand{\suppmtd}{Supplementary Methods}
\newcommand{\supdis}{Supplementary Discussion}

\newcounter{qcounter}

\newcounter{protoCount}
\newcounter{protoList}
\newsavebox{\tmpbox}
\newlength{\protobox}
\newenvironment{protocol}[3]{
\bigskip
\addtocounter{protoCount}{1}
\noindent \begin{lrbox}{\tmpbox}
\setlength{\protobox}{\textwidth}
\addtolength{\protobox}{-9.7 cm}
\begin{minipage}[l]{\protobox}
\begin{bfseries}Protocol #1: #2\end{bfseries}
\ifthenelse{\equal{#3}{\empty}}{}{\\ #3}
\begin{list}{\begin{bfseries}\arabic{protoList}:\end{bfseries}}{\topsep=0in\itemsep=0.08in\parsep=0in\partopsep=0in\usecounter{enumi}}
{\usecounter{protoList}}}
{
\end{list}
\end{minipage}\end{lrbox}
\fbox{\usebox{\tmpbox}}
\bigskip
}

\begin{document}

\title{Experimental implementation of bit commitment in the noisy-storage model}
\author{Nelly Huei Ying Ng}
\affiliation{School of Physical and Mathematical Sciences, Nanyang Technological University, 21 Nanyang Link, 637371 Singapore}
\affiliation{Centre for Quantum Technologies, National University of Singapore, 3 Science Drive 2, 117543 Singapore}
\author {Siddarth K. Joshi}
\affiliation{Centre for Quantum Technologies, National University of Singapore, 3 Science Drive 2, 117543 Singapore}
\author {Chia Chen Ming}
\affiliation{Centre for Quantum Technologies, National University of Singapore, 3 Science Drive 2, 117543 Singapore}
\author{Christian Kurtsiefer}
\affiliation{Centre for Quantum Technologies, 
National University of Singapore,
  3 Science Drive 2, 117543 Singapore}
\affiliation{Physics Department, National University of
  Singapore, 2 Science Drive 3, 117542 Singapore}
\author{Stephanie Wehner}
\email{wehner@nus.edu.sg}
\affiliation{Centre for Quantum Technologies, National University of Singapore, 3 Science Drive 2, 117543 Singapore}
  \affiliation{School of Computing, National University of Singapore, 13 Computing Drive, 117417 Singapore}

\begin{abstract}
Fundamental primitives such as bit commitment and oblivious transfer serve as building blocks for many other two-party protocols. Hence, the secure implementation of such primitives are important in modern cryptography. In this work, we present a bit commitment protocol which is secure as long as the attacker's quantum memory device is imperfect. The latter assumption is known as the noisy-storage model. We experimentally executed this protocol by performing measurements on polarization-entangled photon pairs. Our work includes a full security analysis, accounting for all experimental error rates and finite size effects. This demonstrates the feasibility of two-party protocols in this model using real-world quantum devices. Finally, we provide a general analysis of our bit commitment protocol for a range of experimental parameters.
\end{abstract}
\maketitle

\section{Introduction}
Traditionally, the main objective of cryptography has been to protect communication from the prying eyes of an eavesdropper.
Yet, with the advent of modern communications new cryptographic challenges arose: we would like to enable
two parties, Alice and Bob, to solve joint problems even if they do not trust each other. Examples of such tasks
include secure auctions or the problem of secure identification such as that of a customer to an ATM.
Whereas protocols for general two-party cryptographic problems may be very involved, it is known that they can in principle
be built from basic cryptographic building blocks known as oblivious transfer~\cite{kilian} and bit commitment. 

The task of bit commitment is thereby particularly simple and has received considerable attention 
in quantum information. Intuitively, a bit commitment protocol consists of two phases. In the \emph{commit phase}, Alice
provides Bob with some form of evidence that she has chosen a particular bit $C \in \01$. Later on in the \emph{open phase},
Alice reveals $C$ to Bob. A bit commitment protocol is secure, if Bob cannot gain any information about $C$ before the open phase, 
and yet, Alice cannot convince Bob to accept an opening of any bit $\hat{C} \neq C$. 

Unfortunately, it has been shown that even using quantum communication none of these tasks can be implemented 
securely~\cite{mayers:bitcom, lo:promise, lo:insecurity, lo&chau:bitcom,kretch:bc}. 
Note that in quantum key distribution (QKD), Alice and Bob \emph{trust} 
each other and want to defend themselves against an outsider Eve. This allows Alice and Bob to perform checks on what Eve
may have done, ruling out many forms of attacks. This is in sharp contrast to two-party cryptography where there is no Eve and Alice and Bob \emph{do not trust}
each other. Intuitively, it is this lack of trust that makes the problem considerably harder.
Nevertheless, because two-party protocols form a central part of modern cryptography, one is willing to make \emph{assumptions} on how powerful an attacker
can be in order to implement them securely.

Here, we consider \emph{physical} assumptions that enable us to solve such tasks. In particular, can the sole assumption of a limited storage device lead to security?~\cite{Maurer92b} This is indeed the case and it was shown that security can be obtained if the attacker's \emph{classical} storage is limited~\cite{Maurer92b,cachin:bounded}. 
Yet, apart from the fact that classical storage is cheap and plentiful, assuming a limited classical storage has one rather crucial caveat: If the honest players need
to store $N$ classical bits to execute the protocol in the first place, \emph{any} classical protocol can be broken if the attacker can store
more than roughly $N^2$ bits~\cite{maurer:imposs}. 

Motivated by this unsatisfactory gap, it was thus suggested to assume that the attacker's \emph{quantum} storage 
was bounded~\cite{Bennett84,serge:new,serge:bounded,chris:id1,chris:id2}, or more generally, noisy~\cite{Noisy1, noisy:robust,noisy:new}. The central assumption of the noisy-storage model is that during waiting times $\Delta t$ introduced in the protocol, the attacker can only keep quantum information
in his quantum storage device $\cF$. The exact amount of noise can depend on the waiting time. Otherwise, the attacker may be all-powerful. In particular, he can store an unlimited amount of classical information, and perform any 
computation instantaneously without errors. Note that the latter implies that the attacker could encode his quantum information into an arbitrarily complicated
error correcting code, to protect it from noise in his storage device $\cF$. 

The assumption that storing a large amount of quantum information is difficult is indeed realistic today, as constructing large scale quantum memories that can store 
arbitrary information successfully in the first attempt has proved rather challenging. We emphasize that this model is not in contrast with our ability to build quantum repeaters, where it is sufficient for the latter to store quantum states while making many attempts. A review on quantum memories can be found in \cite{qmemory}, and numerous recent work can also be found in \cite{Usmani2010,bonarota2011,PhysRevLett.108.210501}. While noting that perpetual advances in building quantum memories fundamentally affect the feasibility of all protocols in the noisy storage model, yet we will explain below that given any upper bound on the size and noisiness of a future quantum storage device, security is in fact possible - we merely need to send more qubits during the protocol. 

In this work, we have implemented a bit commitment protocol that is secure under the noisy storage assumption. 
We provide a general security analysis of our protocol for a range of possible experimental parameters. The parameters of our particular
experiment are shown to lie within the secure region. The storage assumption in our work is such that a cheating party cannot store more than approximately 900 qubits, which is a reasonable physical constraint given modern day technology of storing quantum information.
 
\section{Result}

\subsection*{The Noisy Storage Model}
To state our result, let us first explain what we mean by a quantum storage device, and how does an assumption regarding these devices translate to security conditions in the noisy storage model. A more detailed introduction to the model can be found in e.g.~\cite{noisy:new}.

Of particular interest to us are storage devices consisting of $S$ ''memory cells'', each of which may experience some noise $\cN$ itself. 
Mathematically, this means that the storage device is a quantum channel [mathematically, a completely positive trace preserving map (CPTPM)]
of the form $\cF = \cN^{\otimes S}$ where $\cN: \bop(\Complex^d) \rightarrow \bop(\Complex^d)$ is a noisy channel acting on each memory cell mapping input states to some noisy output states.
For example, a noise-free storage device consisting of $S$ qubits (i.e.,$d=2$) corresponding to the special case of bounded storage~\cite{serge:bounded} is given
by $\cF = \id_2^{\otimes S}$ where $\id_2$ is the identity channel with one qubit input and one qubit output. 
Another example is a memory consisting of $S$ qubits, each of which experiences depolarizing noise according to the channel
$\cN_r(\rho) = r \rho + (1-r) \frac{\id}{2}$. The larger $r$ is, the less noise is present. Yet another example is the erasure channel, which models losses in the storage device.

It is indeed intuitive that security should be related to ''how much'' information the attacker can squeeze through his storage device. 
That is, one expects a relation between security and the capacity of $\cF$ to carry quantum information. 
Indeed, it was shown that security can be linked to the classical capacity~\cite{noisy:new}, the entanglement cost~\cite{entCost}, and finally the quantum capacity~\cite{qcextract} of the adversary's storage device $\cF$.

When evaluating security, we start with a basic assumption on the maximum size and the minimum amount of noise in an adversary's storage device.
Such an assumption can for example be derived by a cautious estimate based on quantum memories that are available today. Note that these assumptions are for memories that can store arbitrary states on first attempt. Such memories presently exist for a handful of qubits. Given such an estimate, we then determine the number of qubits we need to transmit during the protocol to effectively overflow the adversary's memory device and achieve security. 

\subsection*{Protocol and its security}

We consider the bit commitment protocol from~\cite{noisy:new} with several modifications to make it suitable for an experimental implementation with time-correlated photon pairs. Figure~\ref{fig:flowchart} provides a simplified version of this modified protocol without
explicit parameters - the explicit version can be found in the~\suppmtd.
In the \suppmtd, we also provide a general analysis that can be used for any experimental setup (details on our particular experiment are also provided in the same section).

To understand the security constraints, we first need to establish some basic terminology.
In our experiment, Alice holds the source, and 
both Alice and Bob have four detectors, each one corresponding to one of the four BB84 states~\cite{Bennett84}.
If Alice or Bob observes a click of exactly one of their detectors (\textit{symmetrized} with the procedure outlined in \suppmtd), we refer to it as a \emph{valid
  click}. Cases where more than one detector clicks at the same instant on the
same side are ignored. 
A \emph{round} is defined by a valid click of \emph{Alice's} detectors.
A \emph{valid round} is where both parties Alice and Bob registered a valid
click in a corresponding time window, i.e., where a photon pair has been
identified.

First, to deal with losses in the channel we introduce a new step in which Bob reports a loss if he did not observe a valid click.
Second, to deal with bit flip errors on the channel, we employ a different class of error-correcting codes, namely a random code. Usage of random codes is sufficient for this protocol since decoding is not required for honest parties. The main challenge is then to link the properties of random codes to the protocol security.

Before we can argue about the correctness and security of the proposed protocol, let us introduce four crucial figures of interest that need to be determined in any experimental setup. The first two are the probabilities $p^0_{\rm sent}$ and $p^1_{\rm sent}$, that none or just a single photon was sent to Bob respectively, conditioned on the event that Alice observed a round.
The third is the probability $p_{\rm B, no click}^{\rm h}$ that honest Bob registers a round as missing, i.e. Bob does not observe a valid click when Alice does.
Again, this probability is conditioned on the event that Alice observed a round. Note that by no-signalling, Alice's choice of better (or worse) detectors should not
influence the probability of Bob observing a round.
Finally, we will need the probability $p_{\rm err}$ of a bit flip error, i.e. the probability that Bob outputs the wrong bit even though he measured in the correct basis. 

Naturally, since Alice and Bob do not trust each other, they cannot rely on each other to perform said estimation process. Note, however, that the scenario
of interest in two-party cryptography is that the honest parties essentially
purchase off the shelf devices with standard properties, for which either of them could perform said estimate. 
It is only the dishonest parties who may be using alternate equipment. Another way to look at this is to say that there exists
some set of parameters (i.e., maximum losses, maxmium amount of noise on the channel, etc) such that an honest party has to conform 
to these requirements when executing the protocol.

Let us now sketch why the proposed protocol remains correct and secure even in the presence of experimental errors. A detailed analysis is provided in the \suppmtd.
In our analysis, we take the storage device $\cF$, as well as a fixed
overall security error $\eps$ as given. Let $M$ be the number of rounds \emph{Alice} registers during the execution of the protocol. 
Let $n$ be the number of valid rounds.
In the description of theoretical parameters found in the \suppmtd , it is shown that $M$ and $n$ are directly related to each other, given some fixed experimental parameters. In particular, $n$ is a function of $M$ and $p_{\rm B, no click}^{\rm h}$
\begin{equation}
n \approx (1 - p_{\rm B, no click}^{\rm h}) M\ .
\end{equation}
We can now ask, how large does $M$ (or equivalently $n$) need to be in order to achieve security. If $n$ is very small, for example if $n\approx 100$, it is relatively easy to break the protocol since a cheating party might be able to store enough qubits. Also many terms from our finite $n$ analysis reach convergence only for sufficiently large $n$. As these terms depend on experimental parameters, security can be achieved for a larger range of experimental parameters if $n$ is large. By fixing the assumption on quantum storage size, experiment parameters and security error values, our analysis allows us to determine a value of $n$ where security is achievable.

\begin{figure}
\includegraphics[width=\columnwidth]{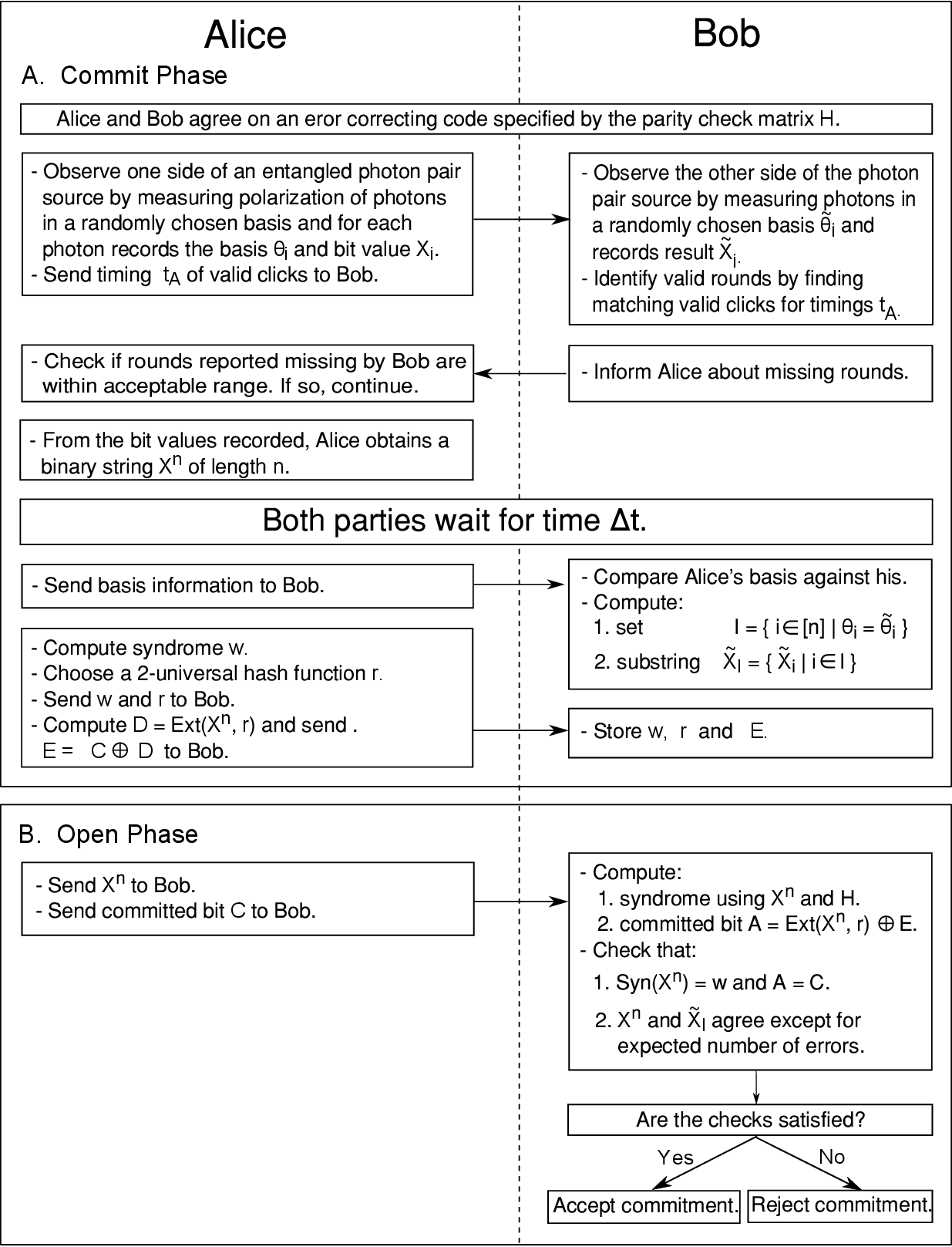}
\caption{\label{fig:flowchart}\small{Flowchart of the bit commitment protocol. This protocol allows Alice to commit a single bit $C\in\01$. Alice holds the source that creates the entangled photon pairs. The function $\Syn$ maps the binary string $X^n$ to its syndrome as specified by the error correcting code. The function $\Ext : \01^n\otimes\cR\rightarrow\01$ is a hash function indexed by $r$, performing privacy amplification. We refer to the \suppmtd ~for a more detailed statement of the protocol including details on the acceptable range of losses and errors.  
Note that the protocol itself does not require any quantum storage to execute.}
}
\end{figure}

{\bf Correctness:} First of all, we must show that if Alice and Bob are both honest, then Bob will accept Alice's honest opening of the bit $C$. Note that the only way that honest Bob will reject Alice's opening is when too many errors occur on the channel, and hence
part 2 of Bob's final check (see Figure~\ref{fig:flowchart}) will fail.  
A standard Chernoff style bound using the Hoeffding inequality~\cite{hoeffding} shows the probability of this event is small, i.e., 
that the deviation from the expected number of $p_{\rm err} n$ errors is not too large. 

{\bf Security against Alice:} Second, we must show that if Bob is honest, then Alice cannot get him to accept an opening of a bit 
$\hat{C} \neq C$. In our protocol, Alice is allowed to be all powerful, and is not restricted by any storage assumptions. If she is 
dishonest, we furthermore assume that she can even have perfect devices and can eliminate all errors and losses on the channel.
The first part of our analysis, i.e., the analysis of the steps before the syndrome is sent is thereby identical to~\cite{Curty10} (see Figure \ref{fig:flowchart}). 
More precisely, it is shown that up to this step in the protocol, a string $X^n \in \01^n$ is generated such that
Bob knows the bits $X_{\mathcal{I}}$ for a randomly chosen subset $\mathcal{I} \subseteq \{1,\ldots,n\}$, where $X_{\mathcal{I}}$ corresponds to the entries of the string $X^n$ indexed by the positions in $\mathcal{I}$. If Alice is dishonest, we want to be sure at this stage
that she cannot learn $\mathcal{I}$, that is, she cannot learn which bits of $X^n$ are known to Bob. 
In the original protocol without experimental imperfections~\cite{noisy:new} this was trivially guaranteed because Bob never 
sent any information to Alice. In this practical protocol, however, Bob does send some information to Bob, namely which rounds are valid for him, i.e., when he saw a click. In~\cite{Curty10} it was simply assumed that the probability of Bob observing a loss is the same for all detectors, and hence in particular also independent of Bob's basis choice. This is generally never the case in practise. 
However, by symmetrizing the losses as outlined in the \suppmtd, one can ensure
that the losses become the same for all detectors. In essence, this procedure probabilistically 
adds additional losses to the better detectors such that in the end all detectors are as lossy as the worst one. 
As Bob's losses are then independent of his basis choice, i.e., the detector, this is means that Alice cannot gain any information about $\mathcal{I}$ when Bob reports some rounds as being lost. 

The second part of the protocol and its analysis uses the string $X^n$ and Bob's partial knowledge $X_{\mathcal{I}}$ to bind Alice to her commitment. First, we have that properties of
the error-correcting code ensure that if the syndrome of the string ($\Syn(X^n)$ in Figure~\ref{fig:flowchart})
matches and Alice passes the first test, then she must flip
many bits in the string to change her mind. In the original protocol of~\cite{noisy:new} sending Bob the syndrome of $X^n$ ensured
that she must change at least $\frac{d}{2}$ bits of $X^n$ where $d$ is the distance of the error-correcting code, 
such that Bob will accept the syndrome to be consistent.
However, since Alice does not know which bits $X_{\mathcal{I}}$ are known to Bob she will get caught
with high probability. This due to the fact that with probability $1-(1/2)^{d/2}$ Alice changed at least a bit known to Bob,
and in the perfect case Bob aborts whenever a single bit is wrong. 
As we have to deal with experimental imperfections we cannot have that Bob aborts whenever a single bit is wrong, as bit flip errors on the channel likely lead errors even when Alice is honest. 
As such the difference to the analysis of~\cite{noisy:new} is that Bob must accept some incorrect bits in part two of 
his final check (see Figure~\ref{fig:flowchart}).
Our argument is nevertheless quite similar, but does require a careful tradeoff involving all experimental parameters 
between the distance of the code and the syndrome length (see below). We hence use a different error-correcting code as compared to~\cite{noisy:new}. In particular, we use a random code which has the property 
that with overwhelming probability its distance is large (i.e. it is hard for Alice to cheat), while nevertheless having a reasonably small syndrome length (see~\supdis). The latter will be important in the security analysis below when Alice herself is honest.


{\bf Security against Bob:} Finally, we must show that if Alice is honest, then Bob cannot learn any information about her bit $C$ before
the open phase. 
Again, dishonest Bob may have perfect devices and eliminate all errors and losses on the channel. His only restriction is that during the waiting time $\Delta t$ he can store quantum information only in the device $\cF$. 

We first show that Bob's information about the entire string $X^n$ is limited.
We know from~\cite{noisy:new} that Bob's min-entropy about the string $X^n$ before Alice sends the syndrome, given all his information including his quantum memory can be bounded by
\begin{align}\label{minEntropy}
\hmin(X^n|\mathrm{Bob}) \gtrsim - \log P_{\rm succ}^{\mathcal{F}}(Rn)\ ,
\end{align}
where $P_{\rm succ}^{\mathcal{F}}(Rn)$ is the maximum probability of transmitting $Rn$ randomly chosen bits 
through the channel $\cF$ where $R$ is called the rate. This rate is determined using a novel uncertainty relation that we prove for 
BB84 measurements, and all experimental parameters. 
The min-entropy itself can thereby be expressed as $\hmin(X^n|\mathrm{Bob}) = - \log P_{\rm guess}(X^n|\mathrm{Bob})$, 
where $P_{\rm guess}(X^n|\mathrm{Bob})$ is the probability that Bob guesses the string $X^n$, maximized over all measurements that he can perform
on his system~\cite{krs:entropy}.  

As Alice sends the syndrome to Bob, Bob gains some additional information which reduces his min-entropy. More precisely, it 
could shrink at most by the length of the syndrome, i.e.,
\begin{align}
\hmin(X^n|\mathrm{Bob},\Syn(X^n)) \geq \hmin(X^n|\mathrm{Bob}) - \log |\Syn(X^n)|\ .
\end{align}
Note that this is the reason why we asked for the error-correcting code to have a short syndrome length above. 

Finally, we show that knowing little about all of $X^n$ implies that Bob cannot learn anything about $C$ itself.
More precisely, when Alice chooses a random two universal hash function $\Ext(\cdot, R)$ and performs privacy amplification~\cite{renato:diss}, Bob knows essentially nothing about the output $\Ext(X^n,R)=D$ whenever his min-entropy about $X^n$ is sufficiently large. 
The bit $D$ then acts as a key to encrypt the bit $C$ using a one-time pad. Since Bob cannot know $D$, he also cannot know $C$. 
Our analysis is thereby very similar to~\cite{noisy:new}, requiring only a very careful balance between the distance of the error-correcting code above, and the syndrome length. 

We provide a detailed analysis in the \suppmtd, where a general statement for arbitrary storage devices is included. Especially for the case of bounded storage $\cF = \id_2^{\otimes S}$, we can easily evaluate how large $M$ needs to be in order to achieve security against both Alice and Bob, when an error parameter $\eps$ is fixed. The total execution error of the protocol is obtained by adding up all sources of errors throughout the protocol analysis.

The case where Alice and Bob are both dishonest is not of interest, because the aim of this protocol is to perform correctly while both players are honest, and protect the honest players from dishonest players.

\subsection*{Experiment}

We have implemented a quantum protocol for bit commitment that is secure in the noisy-storage model. For this, $n = 250\,000$ valid rounds (see below) were used at a bit error rate of $p_{\rm err} = 4.1\%$ (after symmetrization) 
to commit one bit with a security error of less than $\eps = 2 \times 10^{-5}$. Note that $\epsilon$ is the final correctness and security error for the execution of bit commitment in our experiment. This protocol is secure under the assumption that Bob's storage size is no larger than 972 qubits, where each qubit undergoes a low depolarizing noise with a noise parameter $r=0.9$ (see \suppmtd~ Section D).
We stress that our analysis is done for finite $n$, and all finite size effects and errors are accounted for. The $\eps$ includes the error in the choice of random code in the protocol, finite size effects that need to be bounded, smoothing parameters from an uncertainty relation, etc.
Our experimental implementation demonstrates for the first time that two-party protocols proposed in the bounded and noisy-storage models 
are well within today's capabilities.

\section{Discussion}

We demonstrated, for the first time, 
that two-party protocols proposed in the bounded and noisy-storage models can be implemented today.
We emphasize that whereas - like so many experiments in quantum information - our experiment is extremely similar to QKD the experimental parameter requirements and analysis is entirely different to QKD. Where there are many experiments carrying out QKD, there are only a handful of implementation results for two party protocols \cite{nguyen08,berlin11}. Bit commitment is one of the most fundamental protocols in cryptography. For example, it is known that with bit commitment, coin tossing can be built. 
Also using additional quantum 
communication we can build oblivious transfer~\cite{yao:otFromBc}, which in turn enables us to solve any two-party cryptographic problem~\cite{kilian}.
In the \suppmtd, we provided a detailed analysis of our modified bit commitment
protocol including a range of parameters for which security can be shown. Our
analysis could be used to implement the same protocol using a different,
technologically simpler setup, with potentially lower error rates or
losses. Our analysis can also address the case of committing several bits at once. 

It would be interesting to see implementations of other protocols in the noisy-storage model. 

Finally, note that our analysis rests on a fundamental assumption made in in the analysis of 
\emph{all} cryptographic protocols, namely
that Alice does not have access to Bob's lab and vice versa. In particular,
this means that Alice cannot tamper with the random choices made by Bob,
potentially forcing him to measure e.g. only in one basis, or by maniplating
apparent detector losses~\cite{makarov:05,gerhardt:11a}.

\section{Methods}

\subsection*{Parameter ranges}
Our theoretical analysis shows security for a general range of parameters as illustrated in Figures~\ref{fig:securityregion2a},~\ref{fig:securityregion1a} and \ref{fig:securityregion3a}. 
A fully general theoretical statement can be found in the \suppmtd. These plots demonstrate that security is possible for a wide range of parameters, of which our particular implementation forms a special case.
The plots are done for fixed values of $n=250000$ and a total execution error of $\epsilon=3\cdot 10^{-4}$, unless otherwise indicated. 
Finally, Bob's storage size is quantified by $S$, the number of qubits that Bob is able to store. The plots assume a memory of $S$ qubits, where each qubit undergoes depolarizing noise with parameter $r=0.9$.

\begin{figure}[h!]
\includegraphics[width=0.9\columnwidth]{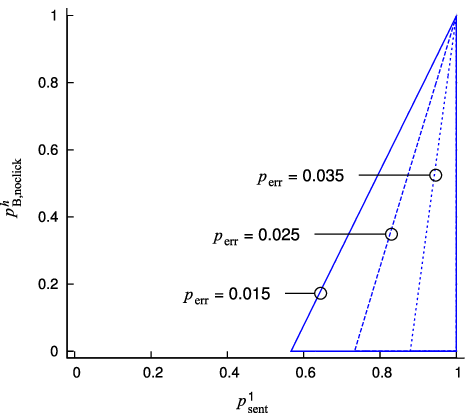}
\caption{\label{fig:securityregion2a}\textbf{Security region for
  $p^1_{\rm sent}$ versus $p^{\rm h}_{\rm B, no click}$.} Plots were done for distinct values of $p_{\rm err}$, while storage
  size is fixed $S=2500$, and $p_{\rm B, no click}^{\rm d} = 0$. For small values of $p_{\rm B, no click}^{\rm h}$ (large amounts of losses), there exists a threshold on $p_{\rm sent}^1$ for the protocol to be secure. This threshold increases with $p_{\rm err}$, 
and for extremely small storage rates, it gives a maximal tolerable $p_{\rm err}\approx 0.046$. }
\end{figure}
\begin{figure}[h!]
\includegraphics[width=0.9\columnwidth]{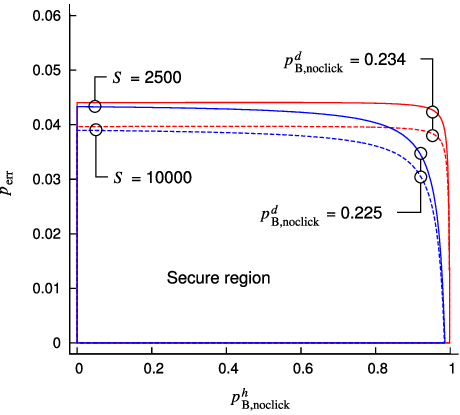}
\caption{\label{fig:securityregion1a} \textbf{Security region for some typical parameter ranges.} $p_{\rm B, no click}^{\rm h}$ and $p_{\rm err}$ quantify the amount of erasures and errors in the protocol. For higher summation values of $p_{\rm B, no click}^{\rm d}+p_{\rm sent}^1$, the less multi-photons Bob gets, and erasures have less impact on the protocol security. This implies if the source is ideal, the protocol remains secure for large values of erasures. Dependences in the security region between erasures and errors also become more obvious when $p_{\rm B, no click}^{\rm d}+p_{\rm sent}^1$ is low. Furthermore, large assumptions on $S$ directly decrease the amount of min-entropy, causing tolerable $p_{\rm err}$ to drop consistently for all amounts of erasures.}
\end{figure}

\begin{figure}[h!]
\includegraphics[width=\columnwidth]{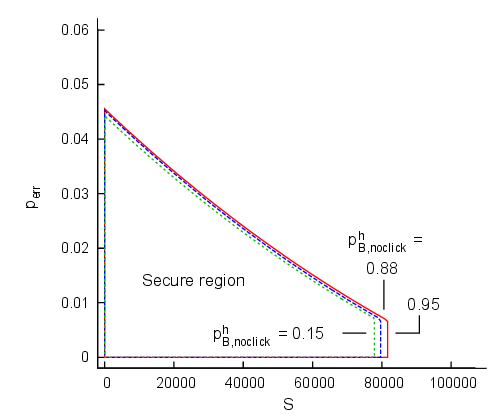}
\caption{\label{fig:securityregion3a}\textbf{Security region for different storage
  size $S$ and error rate $p_{\rm err}$, with  $p_{\rm sent}^1=0.765$, and
  $p_{\rm B, no click}^{\rm d}=0.234$ fixed.} This plot shows a monotonic decreasing trend for tolerable $p_{\rm err}$ w.r.t
storage size $S$. The sharp cut-off for $S$ varies with $p_{\rm B, no click}^{\rm h}$, since with lower detection efficiency, dishonest Bob can report more missing rounds, hence the lower his storage size has to be for security to hold. Also, the plot shows security for mostly low values of storage rate. The result is non-optimal, since it has been shown \cite{entCost} that security can be achieved with arbitrarily large storage sizes, if the depolarizing noise parameter $r \lesssim 0.7$. This is because we bound the smooth min-entropy of an adversarial Bob by the \emph{classical capacity} of a quantum memory, while \cite{entCost} does so in terms of \emph{entanglement cost}. Since the latter is generally smaller than the former, this poses a better advantage for security which is not shown in our analysis. 
}
\end{figure}
\pagebreak
\subsection*{Experimental Implementation}
We implement this protocol with a series of entangled photons, with
the polarization degree of freedom forming our qubits. This allows for
reliable measurements in two complementary bases. Basis 1
corresponds to horizontal/vertical (HV) polarization, and basis 2 to
$\pm45^\circ$ (+-) linear polarization. The polarization-entangled photon pairs are
prepared via spontaneous parametric down conversion (SPDC), collected into
single mode optical fibers, and guided to polarization analyzer (PA) located
with Alice and Bob (see figure~\ref{fig:setup}). Each PA
consists of a non-polarizing beam  splitter (BS) providing a random basis
choice, followed by  two polarizing beam splitters (PBS) and a pair of silicon
avalanche photodiodes (APD) as single photon detectors in each of the BS
outputs. A  half wave plate before one of the PBS rotates the polarization by
45$^\circ$ degrees. This detection setup was used in a number of QKD
demonstrations \cite{kurtsiefer:02b, marcikic:06, ling:08}.

\begin{figure}
\includegraphics[width=\columnwidth]{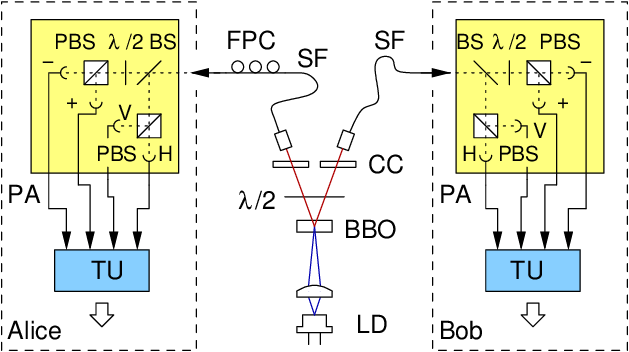}
\caption{\label{fig:setup}
Experimental setup. Polarization-entangled photon pairs are generated
via non-collinear type-II spontaneous parametric down conversion of blue light
from a laser diode (LD) in a barium-betaborate crystal (BBO), and distributed to
polarization analyzers (PA) at Alice and Bob via single mode optical fibers
(SF). The PA are based on a nonpolarizing beam splitter (BS) for a random
measurement base choice, a half wave plate ($\lambda/2$) at one of the of the
outputs, and polarizing beam splitters (PBS) in front of single-photon
counting silicon avalanche photodiodes. Detection events on both sides are
timestamped (TU) and recorded for further processing. A polarization
controller (FPC) ensures that polarization anti-correlations are observed in all measurement bases.} 
\end{figure}

The SPDC source is similar to \cite{ling:08}, with a continuous wave
free running laser diode (398\,nm, 10\,mW) pumping a 2\,mm thick
Barium-betaborate crystal cut for type-II non-collinear parametric down
conversion and the usual walk-off compensation to obtain polarization-entangled
photon pairs \cite{kwiat:95}. We collect  photon pairs into single mode
optical fibers such that we observe an average pair rate $r_\mathrm{p}=2997\pm
82$\,s$^{-1}$.

Such a source generates photon pairs in a stochastic manner, but with a
strong correlation in time. Therefore, valid clicks are timestamped on both
sides first. In a classical communication step, detection times $t_{\rm A},t_{\rm B}$ are
compared, and valid rounds are identified if valid clicks fall into a
coincidence time window of $\tau_\mathrm{c}=3$\,ns, i.e., $|t_\mathrm{A}-t_\mathrm{B}|\le\tau_\mathrm{c}/2$,
similar to \cite{marcikic:06} with the code in \cite{kurtsiefer:08}. The
visibility of the polarization correlations in the Singlet state are
$97.7\pm0.6$\% and $94.7\pm0.9$\% in the HV and $\rm45^\circ$ linear
basis. Individual detection rates on both sides are
$r_\mathrm{A}=23758\pm221$\,s$^{-1}$ and $r_\mathrm{B}=22227\pm247$\,s$^{-1}$ on Alice and Bob's
side, respectively. In an initial alignment step, the fiber 
polarization controller was adjusted such that we see polarization
correlations corresponding to a singlet state with a quantum bit
error ratio (QBER) of about $p_{\rm err}=4.1$\%. The QBER is not to be 
confused with the failure probability of bit commitment protocol. Calculations of 
the latter are explicitly stated in the \suppmtd.
As reported in the summarizing paragraph of our introduction, this quantity is much smaller than the former.

For carrying out a successful bit commitment, we need to determine the
parameters $p^1_{\rm sent}$, $p^0_{\rm sent}$, and $p_{\rm B, no click}^{\rm h}$.
Depending on these probabilities and the desired error parameter $\epsilon$, we
choose a particular error correcting code and number of rounds $M$ needed for a successful bit
commitment. 
To estimate these probabilities out of the
experimental parameters of our source/detector combination, we  model our
setup by a lossless SPDC source emitting only photon pairs at a rate
$r_{\rm s}$, and assign all imperfections (losses, limited detection
efficiency, and background events) to the detectors at Alice and Bob. 
Since the coherence time of
the photons in our case is much shorter than the coincidence detection time
window $\tau_{\rm c}$, the distribution of photon pairs in time can be well
described by a Poisson process, which allows an assessment of multiphoton
events. A detailed derivation of bounds for the probabilities is given in the
\suppmtd, we just summarize the results necessary for
evaluating the security of the protocol:
\begin{eqnarray}
p_{\rm sent}^0&\le&(r_{\rm A}-r_{\rm p})/r_{\rm A}=0.875\pm0.009\,, \\
p_{\rm sent}^{n>1}&<&{r_{\rm A} r_{\rm B}\over r_{\rm p}}\tau_{\rm c}=5.32\pm0.17\times10^{-4}\,,\\
p_{\rm sent}^1&=&1-p_{\rm sent}^0-p_{\rm sent}^{n>1}>0.125\pm0.009\,,\\
p_{\rm sent}^0+p_{\rm sent}^1&=&1-p_{\rm sent}^{n>1}>0.99947\pm0.000017\,,\\
p_{\rm B, no click}^{\rm h}&=&1-r_{\rm p}/r_{\rm A}=0.875\pm0.009\,.
\end{eqnarray}

Due to small differences in the detection efficiency of the APD and
imperfections in polarization components in the actual experiment, there is an
asymmetry in the probability of detecting each bit in
each basis. Furthermore, the beam splitter for the random measurement basis
choice are not completely balanced. A summary of these imperfections over a
number of bit commitment runs is shown in figure ~\ref{fig:asymmetry}.
This can be corrected for by
discarding rounds until the probabilities for both bits are equal. Discarded
bits can be  modeled as losses without affecting the security of the
protocol. A detailed analysis of this can be found in the \suppmtd.

\begin{figure}[h!]
\includegraphics[width=\columnwidth]{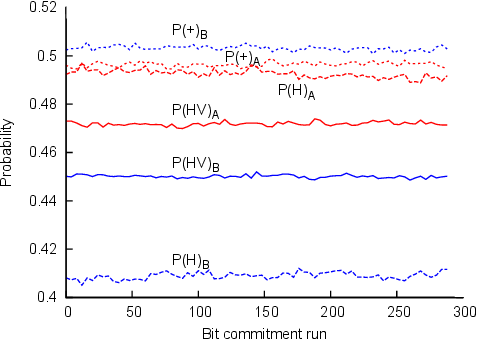}
\caption{\label{fig:asymmetry}Bias in measurements. Solid lines indicate the
  probabilities $P(HV)$ of a HV basis choice for both Alice and Bob for
  data sets of $250000$ events each. Dashed lines indicate the
  probability $P(H)$ of a H in the HV measurement basis, the dotted lines the
  probability $P(+)$ of a $+45^\circ$ detection in a $\pm45^\circ$ measurement
  basis. These asymmetries arise form optical component imperfections and are
  corrected in a symmetrization step.}
\end{figure}
\pagebreak
{\bf Author contributions}
NN, CK and SW designed the research. SJ, CM and CK carried out the experiment. NN wrote the software. NN and SW performed the theoretical analysis. 
NN, SJ, CK and SW wrote the paper. There are no competing financial interests.

\newpage

\setcounter{section}{0}
\renewcommand{\figurename}{\textbf{Supplementary Figure}}
\renewcommand{\tablename}{\textbf{Supplementary Table}}
\setcounter{figure}{0}
\renewcommand{\thefigure}{\textbf{S\arabic{figure}}}
\renewcommand{\thetable}{\textbf{S\arabic{table}}}
\renewcommand{\theequation}{S\arabic{equation}}
\renewcommand{\thetheorem}{\arabic{theorem}}
\setcounter{theorem}{0}
\begin{center}
\large{\textbf{Supplementary Material}}
\end{center}
\vskip-35cm
\section{Supplementary Figures}

\begin{figure}[h!]
\includegraphics[width=0.9\columnwidth]{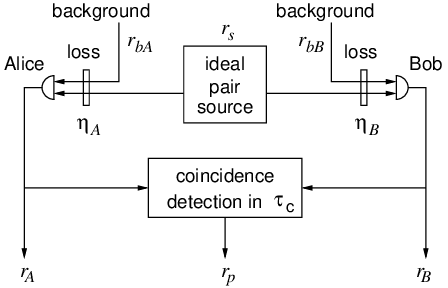}
\caption{\label{fig:model}\textbf{Model of experimental setup.} An ideal source generates time-correlated photon
  pairs with a rate $r_{\rm s}$ and sends them to detectors at Alice and Bob. The losses (due to all causes including source imperfections and detection efficiencies)
  are modeled with attenuators with a transmission $\eta_{\rm A}$ and $\eta_{\rm B}$,
  respectively. To cater for dark counts in detectors, fluorescence background
  and external disturbances, we introduce background rates $r_{\rm bA},r_{\rm bB}$ on
  both sides. Valid rounds are identified  by a coincidence detection
  mechanism that recognizes photons corresponding to a given entangled
  pair. Event rates $r_{\rm A}$ and $r_{\rm B}$ reflect measurable detection rates at
  Alice and Bob, while $r_{\rm p}$ indicates the rate of identified coincidences.}
\end{figure}
\vskip2cm
\begin{figure}[h!]
\includegraphics[width=0.9\columnwidth]{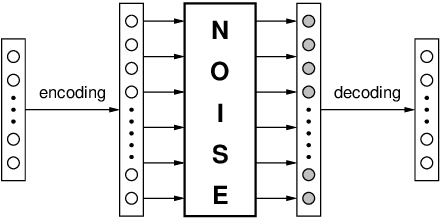}
\caption{\label{fig:channel}\textbf{The encoding and decoding of a message.} A total of \textit{k} bits were encoded into \textit{n} bits and sent through the noisy channel, then recovered completely after undergoing the transmission process.}
\end{figure}

\begin{figure}[h!]
\includegraphics[width=0.9\columnwidth]{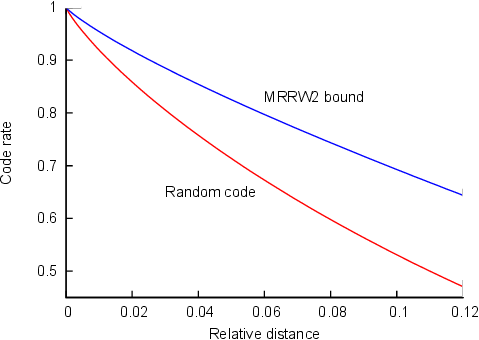}
\caption{\label{fig:reldist}\textbf{Relative distance of random code versus code rate.} A randomly generated code reaches the bound given in Theorem \ref{randomcode} with overwhelming probability. Meanwhile, the MRRW2 bound~\cite{errorcorrcodes} is the smallest upper bound derived up to the present, and it is not known if this bound is tight. It is also not known if there exists any linear binary codes at all between the two regions.}
\end{figure}
\vskip0cm
\pagebreak
\section{Supplementary Tables}
\vspace{-3cm}
\begin{table}[h!]\footnotesize
\begin{tabular}{cll}
\hline ~~~Probabilities~~~~ & Description \\ 
\hline $p^1_{\rm sent}$ & Probability that a single photon was sent to Bob.\\ 
\hline $p_{\rm B,no click}^{\rm h}$ & Probability that honest Bob observes no click.\\ 
\hline $p_{\rm B,no click}^{\rm d}$ & Probability that dishonest Bob observes no
click.\\
& Note: this value is equal to $p_{\rm sent}^0$, i.e., the  \\
& probability that no photons were sent to Bob. \\ 
\hline $p_{\rm err}$ & Probability that the measurement outcome for \\
& honest Alice and honest Bob is different, \\
& when the same basis is used for both parties.  \\ 
\hline 
\end{tabular} 
\caption{\label{probabilities}\textbf{Parameters required for security proof of bit commitment.} All the above quantities are conditioned on the event that Alice registered a valid click.}
\end{table}

\pagebreak

\section{Supplementary Discussion}
\section*{Properties of Error-correcting codes}\label{section:codes}
A linear error-correcting code can be defined by specifying its parity check
matrix H, which has dimensions $n \times (1-R)n$. Given a vector \textbf{x} of length
$n$, the parity check syndrome is simply Syn(\textbf{x})=\textbf{x$\cdot$H}. \cite{errorcorrcodes} \\

Recall that in our BC protocol, both parties agree on a code beforehand, and during the commit phase, Alice sends the syndrome to Bob. The syndrome mainly serves as a back-checking procedure for Bob during the open phase to confirm that Alice is honest. The longer the length of the syndrome, the more information about $X^n$ is given to Bob, and the harder it becomes for Alice to cheat. Both Alice and Bob agree on the code used, and for our purposes the complications in decoding is unnecessary, as an honest Bob never needs to decode.\\

In the theory of error-correcting codes, a question of much significance is depicted in~\suppfig~\ref{fig:channel} : given a scenario where information is sent through an unavoidable noisy channel, under what conditions does an encoding scheme exist such that the message can be recovered completely after undergoing the communication process? In other words, given a message Y comprising of $k$ bits and a noisy channel for communication, what is the theoretical minimum length of encoded message $n$, such that the decoding can detect errors and recover Y accurately?\\

It has been shown by Shannon that for the recovery of information to be possible, the fraction $\frac{k}{n}$ has a theoretical upper bound C, known as the capacity of the channel. For any value $R=\frac{k}{n} > C$, decoding is never possible. For a binary symmetric channel, the capacity is proven to be
\begin{equation}
C_{BSC}(p_{\rm err})=1-h(p_{\rm err}),
\end{equation}
where 
\begin{equation}
{\rm h}(p_{\rm err})= - p_{\rm err}\log_{2}p_{\rm err} - (1-p_{\rm
  err})\log_2(1-p_{\rm err})
\end{equation}
is the binary entropy of the BSC channel.\\

For values of code rate R strictly above the channel capacity, the success probability of delivering the message is exponentially decreasing with code length regardless of the encoding/decoding scheme used.\\

Besides the code rate, another important quantity of error-correcting codes is the minimum distance $d$. Given an error-correcting code, this quantity shows the minimum hamming distance between two strings that have the same parity check syndrome. The larger the minimum distance, the more effective a code is at correcting errors. In the subsequent section, we investigate the relation of parameters $R$ and $d$ for randomly generated codes, and show that random codes satisfy our requirements on these parameters for the protocol to be secure.

\subsection*{Random Codes}
Given a parity check matrix constructed randomly, we are interested in what is the minimum distance of this code. This problem is a computationally NP-hard one, but we do know some probabilistic facts about the minimum distance, which is stated in the theorem below:

\begin{theorem}\label{randomcode}(Random codes, \cite{gallager})
Given a randomly generated binary linear code with rate $R$, the probability that minimum distance $d$ is smaller than some $\delta n$ is bounded by the following:
\begin{equation}
\Pr[d \leq \delta n] \leq 2^{(R-C_\delta)n},\qquad \mbox{for }0\leq\delta\leq 1.
\end{equation}
\end{theorem}

For large block lengths, we can see that this bound approaches a step function where for rates $R<C_\delta$, minimum distance is expected to be larger than $\delta n$ except with extremely small probability. For our choices of block lengths, the randomly generated code will satisfy the bound on minimum distance whenever $R<C_{\delta}$, except for some minimal probability that is later added into the $\epsilon$-error of the protocol. We plot this bound in~\suppfig~\ref{fig:reldist} with respect to the parameter $\delta = \frac{d}{n}$, which we refer to as the relative distance.\\

Given values of $p_{\rm err}$ and reasonably small error parameter
$\epsilon$, by referring to the conditions
for minimum distance derived in the security analysis, we obtain the upper
bound on the achievable rate, namely $C_{\delta}$. This guarantees that for
small enough error rates $p_{\rm err}$, it is sufficient to use a randomly generated code for the use of our protocol, which will provide us both a good enough distance and code rate, except with an extremely small probability. By using Theorem \ref{randomcode} we account for the probability of error and add it as a source of error for the execution of the protocol. \\

\subsubsection*{LDPC codes}
Random binary codes are generated by assigning values 0 and 1 randomly to each element of the parity check matrix. They have a high density (large fraction of non-zero elements), which in large block length limit is time-consuming to deal with. For efficiency purposes, it is of interest whether we can construct codes with lower density (less non-zero values).\\

In \cite{gallager}, Gallager has shown that a specific ensemble of low density codes do attain the same limit given for the random codes as above, when considering large enough block lengths. These codes involve using random permutations of a submatrix, and the construction is straightforward. This type of codes can be of future interest, because their usage will shorten the calculational time used in the protocol. However, this is achieved at the expense of introducing an additional error probability of constructing a bad code (one with unsatisfactory minimum distance), which is not straightforward to evaluate. \\

\subsubsection*{Concatenated codes}
In the case of low bit-flip error, classes of explicit concatenated codes might be generated such that the minimum distance is guaranteed without introducing any probabilistic errors from a randomized construction. These codes are constructed by using a Reed-Solomon code as an outer code, while using a smaller binary linear code as an inner code. We state the properties of such concatenated codes in the following theorem:

\begin{theorem}[Concatenated codes]
Given a $[n_1,R_1n_1,d_1]$ outer code, and a linear binary code with parameters $[n_2,R_2n_2,d_2]$. Then the resulting concatenated code has parameters $[n_1n_2,k_1 k_2, d]$, where the code rate $R=R_1R_2$ and $d \geq d_1 d_2$.
\end{theorem}

For example, by exploiting this construction, a linear binary concatenated
code with rate R=0.53 and relative minimum distance $\delta\geq0.052$ can be
constructed, where the code length $n=311296$. This value of $\delta$ has a
large discrepancy compared to the probabilistic argument for a random
code. From here it is clearly shown that, if a definite statement regarding
the minimum distance of such large error-correcting codes (without any
probabilistic errors) is desired, one can still obtain security for smaller
ranges of experimental parameters. For the given example of concatenated RS
code, this corresponds to security for bit flip error rates $p_{\rm err}\leq 0.02$, which exceeds the value obtained in our experiment.

\section{Supplementary Methods}
\subsection{Experimental parameters\label{section:expparams}}

To analyze our bit commitment protocol in any practical experiment, several probabilities have to be determined. 
\supptab~\ref{probabilities} summarizes all the probabilities we will need to estimate. We emphasize that all such probabilities are conditioned on the event that Alice registers a round, i.e. sees a valid click.

A difficulty in estimating the probabilities of success in a ``round'' arises
from the fact that generation of photon pairs in a parametric down
conversion source is a stochastic process. Furthermore, losses in the system
may occur in the source or in detectors, and we do not have an easy way of
assessing the losses reliably. We thus try to estimate bounds of the required
probabilities for the bit commitment protocol out of observable quantities
both Alice and Bob can agree upon. For this purpose, we model losses and
background events in our system in a way shown in \suppfig~\ref{fig:model}.

The rates (i.e., events per unit of time) observed at Alice are then given by
\begin{equation}
r_{\rm A}=\eta_{\rm A}(r_{\rm s}+r_{\rm bA})\,,\label{eq:singlesAlice}
\end{equation}
where $\eta_{\rm A}$ indicates the detection efficiency
and $r_{\rm bA}$ a background event rate; a similar expression holds for Bob. 
The observed coincidence rate in this model is given by
\begin{equation}
r_{\rm p}=\eta_{\rm A}\eta_{\rm B}r_{\rm s}+r_{\rm acc}\,.
\end{equation}
where $r_{\rm acc}$ reflects the so-called accidental coincidence rate, caused
by detection events on both sides happening within the coincidence time window
$\tau_{\rm c}$ that are not due to valid clicks form the same photon
pair. This rate can be bounded from observed rates $r_{\rm A}$ and $r_{\rm B}$ to 
\begin{equation}
r_{\rm acc}<r_{\rm acc}^{\rm max}=r_{\rm A}r_{\rm B}\tau_{\rm c}\,,
\end{equation}
assuming that all detection events on both sides are caused by uncorrelated
events. In our experiment, this quantity would result in a value of $r_{\rm
  acc}^{\rm max}=14.9\pm 0.18$\,s$^{-1}$, and is negligible compared to the
observed coincidence rate $r_{\rm s}$. This quantity was independently assessed by
recording the rate of detection time pairings $t_{\rm A},t_{\rm B}$ in a time window that
was displaced by $\tau_{\rm d}=20$\,ns from the ``true'' coincidences, i.e.,
$|t_{\rm A}-t_{\rm B}-\tau_{\rm d}|\le\tau_{\rm c}$ \cite{marcikic:06}. We 
found a rate of $r_{\rm acc}=5.3\pm3.3$\,s$^{-1}$ over the course of several
bit commitment runs. Since $r_{\rm acc}\ll
r_{\rm p}$, we from now on neglect these events in the rate estimations, and
interpret their occurence just as events that increase the error ratio.

To evaluate the probability $p^1_{\rm sent}$ that exactly one photon was sent
to Bob in the interval $\tau_{\rm c}$ around a time when Alice has seen an
event, we first consider the probability $p^0_{\rm sent}$ that no photon was
sent to Bob, given Alice has seen an event. This can only be caused by a
background event with Alice. Thus, $p^0_{\rm sent}$ equals the probability
that a detection event on Alice's side is caused by background, which 
is given by 
\begin{eqnarray}
p^0_{\rm sent}&=&{r_{\rm bA}\over r_{\rm bA}+r_{\rm s}}=1-{r_{\rm s}\over r_{\rm bA}+r_{\rm s}}\nonumber \\
&=&1-{\eta_{\rm A} r_{\rm s}\over\eta_{\rm A}(r_{\rm bA}+r_{\rm s})}=1-{\eta_{\rm A}r_{\rm s}\over r_{\rm A}}\nonumber \\
&=&1-{r_{\rm p}\over\eta_{\rm B} r_{\rm A}}\label{eq:p0sent}
\end{eqnarray}
Since the efficiency $\eta_{\rm B}$ is not known exactly, we set it to 1 and thereby
obtain an upper bound for  $p^0_{\rm sent}$:
\begin{equation}
p^0_{\rm sent}<1-{r_{\rm p}\over r_{\rm A}}=0.875\pm0.009
\end{equation}

Next, we consider the probability $p^{n>1}_{\rm sent}$ that more than one
photon has been sent to Bob, given that Alice has seen an event. This
probability is the product of the probability 
that Alice's event was caused by a photon pair, and the probability that at
least one other photon pair than the one causing the event on Alice's side was
generated in the coincidence time window $\tau_{\rm c}$.
From equation~\ref{eq:p0sent}, the first probability is given by
$r_{\rm p}/(\eta_{\rm B}r_{\rm A})$. For the latter, we consider the statistics of photon pairs
emerging from a continuously pumped SPDC source. While light emerging from a
downconversion process is known to follow thermal photon counting statistics,
the coherence 
time of the photons in our case (0.73\,ps for an optical bandwidth of
3\,nm) is much shorter than $\tau_{\rm c}$. In this case, the statistics of several
photon pairs in time window $\tau_{\rm c}$ follows a Poisson distribution. Since the
creation of an additional photon pair is then independent of the
first photon pair, and the probability that no photon pair is created in
$\tau_{\rm c}$ is given by $e^{-r_{\rm s}\tau_{\rm c}}$, the probability of creating at least
one more photon pair is given by $1-e^{-r_{\rm s}\tau_{\rm c}}$. This brings us to
\begin{eqnarray}
p^{n>1}_{\rm sent}&=&{r_{\rm p}\over \eta_{\rm B}r_{\rm A}}(1-e^{-r_{\rm s}\tau_{\rm c}})\nonumber \\
&<&{r_{\rm p}\over \eta_{\rm B}r_{\rm A}}r_{\rm s}\tau_{\rm c} = 
{r_{\rm p}\over \eta_{\rm B}r_{\rm A}}{r_{\rm p}\over \eta_{\rm A}\eta_{\rm B}}\tau_{\rm c}\nonumber \\
&=&{r_{\rm p}^2\over r_{\rm A}\eta_{\rm A}\eta_{\rm B}^2}\label{eq:pg1interim}
\end{eqnarray}
The efficiencies $\eta_{\rm A}$, $\eta_{\rm B}$ are not accessible directly from the
experiment, but can be bounded by $\eta_{\rm A}>r_{\rm p}/r_{\rm B}$ and $\eta_{\rm B}>r_{\rm p}/r_{\rm A}$ via
\ref{eq:singlesAlice}.
With this, we can further bound expression~\ref{eq:pg1interim} and arrive at
\begin{equation}
p^{n>1}_{\rm sent}<{r_{\rm A}r_{\rm B}\over r_{\rm p}}\tau_{\rm c} = 5.32\pm0.17\times10^{-4}\,,
\end{equation}
which is much smaller than the uncertainty on $p^0_{\rm sent}$.  With this, we
arrive at 
\begin{equation}
p^1_{\rm sent}=1-p^0_{\rm sent}-p^{n>1}_{\rm sent}>0.125\pm0.009
\end{equation}
and
\begin{equation}\label{a9}
p^1_{\rm sent}+p^0_{\rm sent}=1-p^{n>1}_{\rm sent}>0.99947\pm0.000017
\end{equation}

Finally, the probability for an honest Bob not seeing an event in a
coincidence time window if Alice has detected something is the complement to
the probability that Bob sees something if Alice has seen something. The
latter, by definition, is given by the ratio $r_{\rm p}/r_{\rm A}$. Thus, we have
\begin{equation}
p_{\rm B, no click}^{\rm h} = 1-r_{\rm p}/r_{\rm A} = 0.875\pm0.009\,.
\end{equation}

\subsection{Symmetrizing losses}\label{sec:symLoss}

In practice, not all detectors have the same efficiency. Losses will be higher for some detectors than for others. This will lead to imbalances in the choice of basis and the choice of BB84 encoded qubit. In our protocol, such imbalances affect the security in two places. First, if Alice is honest, but Bob is trying to cheat, such imbalances give him additional information about which bit or basis was used. His advantage is similar to the advantage that an eavesdropper in QKD would gain from knowing such extra information. Second, if Bob is honest, but Alice is trying to cheat, having higher losses in one basis does reveal information to Alice in which basis Bob measured - if Bob does not report a loss it is more likely that he used the basis for which losses occur less often.

We describe a method to deal with such imbalances securely - the same method can be used to address imbalances on Alice's and Bob's side. For simplicity, we outline the procedure in detail for Alice; exactly the same method can be used to symmetrize Bob's detectors.
The essential idea is to make all detectors equally inefficient, by throwing away (i.e., declaring as lost) rounds where detectors with higher efficiencies registered a click. Note that in our protocol, Alice can discard additional rounds without consequences for security parameters. Meanwhile, discarding additional rounds on Bob's side increases $p_{\rm B, no click}^{\rm h}$. Detection events combining with such post-processing procedures, define the occurrence of a valid round. In other words, if a single click occurred on both sides and was not manually discarded for symmetrizing purposes, this event is considered a valid round.

In our setup, Alice has four detectors, one for each bit in each basis. Let $x,\theta$ label the detector corresponding to a bit $x \in \01$ in basis $\theta \in \01$. 
Let $p_{\theta}$ denote the probability that basis $\theta$ is chosen, and let $p_{x|\theta}$ denote the probability that bit $x$ occurs given basis $\theta$. Finally, let
$t_{x,\theta}$ denote the probability that Alice keeps bit $x$ in basis $\theta$ when the detector $x,\theta$ clicks. That is, Alice discards bit $x$ in basis $\theta$ with probability 
$1-t_{x,\theta}$ even though a click occurred. Our goal will be to determine the $t_{x,\theta}$ that renders $\Pr[x,\theta|{\rm keep}]$, the probability that $x,\theta$ occurs conditioned on the event that Alice keeps a particular detection event the same for all $x$ and $\theta$.

First of all, note that the probability that a particular detection event is \emph{not} discarded, i.e. Alice accepts it as a round, can be written as
\begin{align}
	\Pr[{\rm keep}] = \sum_{x,\theta \in \01} p_{\theta} p_{x|\theta} t_{x,\theta}\ .
\end{align}
By Bayes' rule
\begin{align}
	\Pr[x,\theta|{\rm keep}] &= \frac{\Pr[{\rm keep}|x,\theta] \Pr[x,\theta]}{\Pr[{\rm keep}]}\\
	&= \frac{t_{x,\theta} p_{x|\theta} p_{\theta}}{\Pr[{\rm keep}]}\ .
\end{align}
Ideally, all probabilities are the same, i.e., for all $x$ and $\theta$
\begin{align}
	\Pr[x,\theta|{\rm keep}] = \frac{1}{4}\ .
\end{align}
This yields $4$ equations, in $3$ free parameters since $\sum_{x,\theta} t_{x,\theta} = 1$. These can easily be solved
for $t_{x,\theta}$. 

For our setup, the parameters for symmetrization on Alice' side are as follows:
\begin{eqnarray}
&& t_{0,0} = 1 \nonumber\\
&& t_{0,1} = 0.963077 \nonumber\\
&& t_{1,0} = 0.882305 \nonumber\\
&& t_{1,1} = 0.871353.
\end{eqnarray}
Symmetrization on Bob's side is dealt with in the same manner. This
however increases the value of $p_{\rm B, no click}^{\rm h}$, since now an honest Bob deliberately throws away more clicks. This leads to a new value of
\begin{equation}
\tilde{p}_{\rm B, no click}^{\rm h} = 1 - (1-p_{\rm B, no click}^{\rm h})\cdot \Pr\left[{\rm keep}\right].
\end{equation}
For our setup, the parameters for symmetrization on Bob's side are:
\begin{eqnarray}
&& t_{0,0} = 0.679745 \nonumber\\
&& t_{0,1} = 1 \nonumber\\
&& t_{1,0} = 0.665591 \nonumber\\
&& t_{1,1} = 0.662890.
\end{eqnarray}

The probability of Bob keeping a click during symmetrization is $\Pr\left[{\rm keep}\right]=0.729646$. This combining with the initial estimate of $p_{\rm B, no click}^{\rm h}$ gives $\tilde{p}_{\rm B, no click}^{\rm h}=0.909$, implying a high amount of losses. Even so, the protocol remains secure due to the fact that the source provides multi-photons to Bob with an extremely small probability, whenever Alice observes only a single detection event. In other words, $p_{\rm sent}^1 + p_{\rm B, no click}^{\rm d}$ is extremely high, as stated in \eqref{a9}. In such cases, even a high amount of losses do not compromise security of the protocol.

Also, it should be stressed that $p_{\rm err}$ should be evaluated for the set of data after all symmetrization procedures, since there can be bias in the error rates for each bit and basis. For the set of symmetrized data, $p_{\rm err}=0.0412$, in comparison with before symmetrization, $p_{\rm err}=0.0428$.

\subsection{Theoretical security analysis}\label{sec:analysis}

In the security proof, we divide the protocol into two parts: the first part is Weak String Erasure with Errors (WSEE), and the remaining procedure is Bit Commitment (BC). 

\subsection*{Theoretical parameters}\label{theoreticalparam}
Next to the experimental parameters defined in the \supptab~\ref{probabilities}, our analysis will make frequent use of the following parameter definitions. There are two more basic parameters in this analysis: $M$ and $\epsilon$.

The parameter $\epsilon$ represents a fixed error parameter, i.e., we want to achieve security up to an error of $O(\eps)$. This parameter is used to bound the occurrence probability of bad events, and we need to frequently refer to it throughout the analysis. 
Such bounds are achieved by making use of the Hoeffding inequality. It says that given a random variable $X_{j} \in \lbrace0,1\rbrace$, where $\Pr(X_{i}=0)=1-p$, $\Pr(X_{i}=1)=p$, and
$Y=\Sigma_{i=1}^{N} X_{i}$ we have
\begin{equation}
\Pr [ Y \leq (p-\alpha)N] = \Pr [Y \geq (p+\alpha)N] = e^{-2\alpha^{2}N}.
\end{equation}
The way we will use the Hoeffding inequality is that we demand that $e^{-2 \alpha^2 N} \leq \eps$, and then solve for $\alpha$ such that
our demand is satisfied.

Meanwhile, $M$ denotes the number of signals that Alice counted as valid, i.e., she registers a round (but not necessarily Bob as well).

Based on $\epsilon$ and $M$, we will need the following definitions:
\begin{align}\label{param}
\zeta_{\rm B, no click}^{\rm h} &:= \sqrt{\frac{\ln \frac{2}{\epsilon}}{2M}}\nonumber\\
n &:= (1-p_{\rm B, no click}^{\rm h}-\zeta_{\rm B, no click}^{\rm h}) M \nonumber\\
\alpha_1 &:=\sqrt{\frac{\ln\frac{1}{\epsilon}}{2n}}\nonumber\\
m &:= \left(\half - \alpha_1\right)n\nonumber\\
\alpha_2 &:= \sqrt{\frac{\ln\frac{2}{\epsilon}}{2m}}\nonumber\\
\alpha_3 &:= \sqrt{\frac{\ln\frac{1}{\epsilon}}{d}}
\end{align}
where $d$ is the minimum distance of the error-correcting code used in the protocol, and $n$ is the number of valid rounds that remain.

\subsection*{Weak string erasure}

\subsubsection*{Definition}
We first provide an informal definition of weak string erasure with errors (WSEE). A formal definition can be found in~\cite{noisy:new}. 
When both Alice and Bob are honest, an $(n,\lambda,\epsilon,p_{\rm err})$-WSEE scheme provides Alice with a string $X^n$ and Bob with a 
randomly chosen subset $\cI\in [n]$, as well as a substring $\tilde{X}_{\mathcal{I}}$. 
This substring is thereby given by the substring $X_{\mathcal{I}}$ (the bits of $X^n$ corresponding to the indices in $\mathcal{I}$) passed through a binary symmetric channel that flips each bit of $X_{\mathcal{I}}$ with probability $p_{\rm err}$.

To specify the security condition against dishonest Bob, we first need to 
quantify the uncertainty of Bob about $X^n$, given access to the entire system of a dishonest Bob denoted as ${\rm B'}$. This is done by lower bounding the min-entropy of $X^n$ conditioned on Bob's information,
\begin{eqnarray}
\hmin(X^n|{\rm B'})_{\rho_{X^n{\rm B'}}}&:=&-\log P_{\rm guess} (X^n|{\rm B'}) \nonumber\\[0pt]
P_{\rm guess}(X^n|{\rm B'}) &:=& \max_{ \lbrace D_x \rbrace_x} \displaystyle\sum_x P_X(x) \tr(D_x\rho_x),
\end{eqnarray}
where $P_{\rm guess}$ is referred to as the guessing probability, namely the probability that Bob correctly guesses $X^n$, maximized over all measurement strategies upon his system ${\rm B'}$~\cite{krs:operational}. 
The \textit{$\epsilon$-smooth min-entropy} is defined as 
\begin{equation}
\hmineps (X^n|{\rm B'})_{\rho_{X^n{\rm B'}}} := \sup_{\rho'}~\hmin(X^n|{\rm B'})_{\rho'}
\end{equation}
maximized over all states $\rho'$ such that the purified distance $C(\rho',\rho_{X^n{\rm B'}})= \sqrt{1-F^2 (\rho',\rho_{X^n{\rm B'}})} \leq \epsilon$, where $F(\rho,\tau)$ denotes the fidelity of states $\rho$ and $\tau$. Intuitively, this quantity behaves like the min-entropy, except with a probabilistic error $\epsilon$. \\

We can now state the security conditions for an $(n,\lambda,\epsilon,p_{\rm err})$-WSEE:\\[-4pt]

\noindent{\bf 1. Security for Alice:} If Alice is honest, then the amount of information a dishonest Bob holds about $X^n$ is limited, i.e. the $\epsilon$-smooth min entropy of $X^n$ conditioned on a dishonest Bob's information is lower bounded
\begin{equation}
\frac{1}{n}~\hmineps (X^n|{\rm B'}) \geq \lambda,
\end{equation}
where $\lambda$ is referred to as the smooth min-entropy rate.\\[-4pt]

\noindent{\bf 2. Security for Bob:} If Bob is honest, then Alice does not have any information $\cI$. That is, Alice does not learn which bits of $X^n$ are known to Bob.

\subsubsection*{Protocol}

In principle, WSEE can be achieved experimentally by using any QKD device. However, we emphasize that the 
experimental requirements and analysis differs entirely. In particular, security of QKD for a particular setup does not imply security of bit commitment. 

Recall from the informal statement of the protocol in the main part of our paper that if Alice herself concludes that no photon or a multi-photon
has been emitted in a particular time slot, she simply discards this event and tells
Bob to discard it as well. Since this action represents no security problem for us, we will
for simplicity omit these events all-together when stating the more detailed protocol below. 
This means that $M$ in the protocol below, actually refers to the set of
post-selected pulses that Alice did register as a round.
In practice, Alice reports
the missing events to Bob after the waiting time has passed. In principle, this could be used to obtain better security bounds
as Bob does not yet know which bits are indeed relevant when he uses his storage device. However, we leave such a refined analysis for future work.

In addition, introducing time slots enables Bob to report a
particular bit as missing, if he obtained no click in a particular
time slot. Alice and Bob will subsequently discard all lost rounds.
In the protocol below, we assume the detectors have already been symmetrized appropriately as outlined in the \suppmtd~\ref{sec:symLoss}. The purpose of symmetrizing is to ensure that losses are independent of basis choice, hence Alice cannot obtain any information about $\cI$ by observing the rounds reported lost by Bob.

\begin{protocol}{1}{Weak String Erasure with Errors (WSEE)}
{Outputs: $x^n \in \sbin^n$ to Alice, 
$(\cI,z^{|\cI|}) \in 2^{[n]} \times \sbin^{|\cI|}$ 
to Bob.}\label{proto:wse}
\item[1.] {\bf Alice:} Chooses a string $x^M \in_R \01^M$ and basis-specifying 
string $\theta^M \in_R \01^M$ uniformly at random. 
\item[2.] {\bf Bob: } Chooses a basis string $\tilde{\theta}^M \in_R \01^M$ uniformly at random. 
\item[3.] In time slot $i=1,\ldots,M$: 
\begin{enumerate}
\item {\bf Alice:}
Encodes bit $x_i$ in the basis $\theta_i$ (i.e., as $H^{\theta_i}\ket{x_i}$), and sends
the resulting state to Bob.
\item {\bf Bob:} Measures in the basis
given by $\tilde{\theta}_i$ to obtain outcome $\tilde{x}_i$.
If Bob obtains no click in this time slot, he records round $i$ as lost.
\end{enumerate}

\item[4.] {\bf Bob: } Reports missing rounds to Alice.

\item[5.] {\bf Alice: } If the number of rounds that Bob reported missing does not lie in the 
interval $[(p_{\rm B, no click}^{\rm h} - \zeta_{\rm B, no click}^{\rm h})M,(p_{\rm B, no click}^{\rm h} + \zeta_{\rm B, no click}^{\rm h})M]$, then Alice aborts 
the protocol. Otherwise, she deletes all bits from $x^M$ that Bob reported missing. Let $x^n \in \01^n$
denote the remaining bit string, and let $\theta^n$ be the basis-specifying string for the remaining rounds.
Let $\tilde{\theta}^n$, and $\tilde{x}^n$ be the corresponding strings for Bob.

\item[Both parties wait time $\Delta t$.]

\item[6.] {\bf Alice: } Sends the basis information $\theta^n$ to Bob, and outputs $x^n$.
\item[7.] {\bf Bob: } Computes $\cI := \{i \in [m] \mid \theta_i = \tilde{\theta}_i\}$, and outputs $(\cI,z^{|\cI|}):=(\cI,\tilde{x}
_{\cI})$.
\end{protocol}

How large is $n$ going to be? Since Alice aborts if Bob reports too many
rounds as missing, we have that $n \geq (1 - p_{\rm B, no click}^{\rm h} -
\zeta_{\rm B, no click}^{\rm h}) M$.
If a fixed $n$ is desired, we can take $n =  (1 - p_{\rm B, no click}^{\rm h} -
\zeta_{\rm B, no click}^{\rm h}) M$ as in~\eqref{param}, 
where Alice randomly truncates the resulting string, and informs Bob about the truncation. This is the approach we take here.
In our protocol, there is also a possibility that Alice aborts. An abort here means that Alice simply generates a random $x^n$ as output.
This means that our protocol does at all times generate a string of length $n
= (1 - p_{\rm B, no click}^{\rm h} - \zeta_{\rm B, no click}^{\rm h}) M$.

\subsubsection*{Analysis}

The analysis of weak sting erasure with errors has already been performed in~\cite{Curty10}. 
Essentially, losses allow a dishonest Bob to discard a fraction of single-photon detection events, and keep more multi-photon events so that his chance of guessing $X^n$ correctly is increased. 
The resulting min-entropy rate $\lambda$ can thereby be calculated as a function of experimental parameters listed in \supptab~ \ref{probabilities}. 
That is, the min-entropy rate is a function of $p_{\rm sent}^1$, $p_{\rm B, no click}^{\rm h}$,  $p_{\rm B, no click}^{\rm d}$, and $M$
\begin{align}
\lambda = \lambda(p_{\rm sent}^1,p_{\rm B, no click}^{\rm h}, p_{\rm B, no click}^{\rm d},M)\ .
\end{align}
In our further analysis, we will hence assume that WSEE has been shown secure, and state security as a function of a fixed 
parameter $\lambda$ and fixed $n$. We then combine the two to give explicit security paramters as a function of the 
experimental parameters.

\subsection*{Bit commitment from weak string erasure}

What remains is to analyse security of the bit commitment protocol(BC) based on WSEE. An informal definition of BC was stated in the introduction, and a formal one can be found in~\cite{noisy:new}.
The protocol below is very similar to the one proposed in~\cite{noisy:new}, which gave a BC protocol for weak string erasure \emph{without} errors (i.e., $p_{\rm err} = 0$). 
To address the case of $p_{\rm err} > 0$, we introduce modifications to the BC protocol, allowing the modified protocol to stay secure up to a certain amount of bit flip error from the experimental setup. 

\subsubsection*{Protocol}

We present the fully modified BC protocol as Protocol 2, by including WSEE as a sub-protocol. Our protocol allows Alice to commit
a string $D^l \in \01^l$ to Bob. However, for our experiment we chose to commit only a single bit $l=1$ which is the scenario typically considered in bit commitment. Our protocol makes use of the parameters defined in~\eqref{param}.

\begin{protocol}{2}{Non-Randomized Bit Commitment(BC)}
{By using an binary linear error-correcting code $\mathcal{C}$, let Syn:$\lbrace0,1\rbrace^n \rightarrow \lbrace0,1\rbrace^{n-k}$ be the function that outputs the parity check syndrome for $\mathcal{C}$.
Also, select Ext: $\lbrace0,1\rbrace^n \times \mathit{R} \rightarrow \lbrace0,1\rbrace ^l$ from a set of 2-universal hash functions. }\label{proto:bc}
\item[(A)]\textbf{Commit Phase}
\item[1.] {\bf Alice and Bob :} execute $(n,\lambda,\epsilon,p_{\rm err})$WSEE. Alice obtains $X^n$ while Bob obtains $\tilde{X}_\mathcal{I}$ and $\mathcal{I}$.
\item[2.] {\bf Bob :} Checks if $|\cI| \geq m$. If so, he randomly truncates $\cI$ until $|\cI|=m$. Otherwise, he aborts the protocol.
\item[2.] {\bf Alice:} \\
a) computes w=Syn($X^n$) and sends it to Bob.\\
b) picks a 2-universal hash function $r \in_R \mathcal{R}$ and sends it to Bob.
\item[3.] {\bf Alice:} Commits $C^l \in \lbrace0,1\rbrace^l$ by computing $D^\mathit{l}=\Ext(X^n,r)$ and sends $E^l=C^l \oplus D^l$ to Bob.
\item[(B)]\textbf{Open Phase}
\item[1.] {\bf Alice:} reveals the complete string $X^n$ to Bob.
\item[2.] {\bf Bob:} Perform checks:\\
	a) computing the syndrome and check that it agrees with $w$ sent by Alice.\\
	b) checking $X^n$ against $\tilde{X}_\mathcal{I}$, ensuring that the number of bits that disagree at positions in $\cI$
	lie in the interval [$(p_{\rm err}-\alpha_2)m,(p_{\rm err}+\alpha_2)m$].
\item[3.] {\bf Bob:} If conditions are satisfied, he accepts commitment and calculates $D^\mathit{l} = \Ext (X^n,r)$. Both of them output $C^l$.
\end{protocol}

The are two modifications in this protocol compared to its previous version in \cite{noisy:new}. The first is the use of a different error-correcting code $\mathcal{C}$. We will first provide a general analysis to prove the security of bit commitment, given that several conditions on both the rate and relative minimum distance of the error-correcting code used are met. We then show that by generating a binary linear code at random, the conditions on minimum distance and rate can be satisfied, except for a small probabilistic error that can be later added upon the total security error of the protocol. 

Note that random codes do not pose a problem when executing the protocol since honest parties never need to decode. Details of the properties of error-correcting codes can be found in the \supdis. Secondly, the checks performed by Bob in the open phase have been modified to account for the existence of bit flip errors, such that Bob tolerates a certain limited amount of errors.

\subsubsection*{Analysis}

Intuitively, bit flip errors in WSEE give Alice more freedom to cheat, as a
malicious Alice can avoid the errors and choose to corrupt the string $X^n$
herself. In other words, the actual bit flip error $p_{\rm err}$ in such a scenario equals zero. This makes it harder for Bob to identify a cheating Alice. This is because whenever he finds a discrepancy between $X_\mathcal{I}$ and $X^n$, he cannot be sure if it was due to a bit flip error or a malicious Alice.\\

We now proceed to prove that the protocol is secure, except for a minimal probability $\epsilon$ when $p_{\rm err}$ is sufficiently small. The proof is done in three steps as shown:
\begin{list}{\arabic{qcounter}:~}{\usecounter{qcounter}}
\item[1. Correctness:] If both parties are honest, Bob always accepts the commitment except with minimal probability.
\item[2. Security against Alice:] For any dishonest Alice, Bob detects her attempt to cheat and rejects the commitment except with minimal probability.
\item[3. Security against Bob:] For any dishonest Bob, he does not obtain information about the committed bit except with minimal probability.
\end{list}

\paragraph{Correctness}

We shall first prove the correctness of the protocol, namely under the situation that Alice and Bob are both honest, Bob always accepts the protocol except with some minimal probability $\epsilon$.

\begin{lemma}[Correctness of the protocol]\label{correctness}
If Alice and Bob are both honest, then the protocol is $2\epsilon$-correct.
\end{lemma}

\begin{proof}
There are two steps in this proof. Firstly, we show that Bob receives at least
$m$ bits from WSEE except with probability $\epsilon$. Secondly, we prove that
the number of erroneous bits Bob picks up is close to the expected value
$p_{\rm err} m$, except for probability $\epsilon$. The total probability of error for either events occurring is then $2\epsilon$ since failure occurs in either case.\\

Since each bit $X_i$ from Alice is obtained by Bob with probability $\frac{1}{2}$, by applying the Hoeffding's inequality to the random variable $Y=|\mathcal{I}|$, i.e. the length of substring Bob obtains from WSEE, we see that
\begin{equation}
	\Pr [ Y \leq (\frac{1}{2}-\alpha_1)n] \leq e^{-2\alpha_{1}^{2}n} \leq \epsilon,
\end{equation}
where n is the length of string $X^n$ Alice has. Using $\alpha_1$ defined in \eqref{param} allows Bob to get at least $m$ bits except for probability $\epsilon$. Note that the probability $1/2$ is determined by Bob's random choice of basis, and is independent 
of $p_{\rm err}$.\\

We then proceed to show that the number of erroneous bits Bob expects to
obtain lie within the interval $[(p_{\rm err}-\alpha_2)m,(p_{\rm err}+\alpha_2)m]$ except for some probability $\epsilon$. Note that since previously we have shown that Bob obtains at least $m$ bits, we can safely fix the length of Bob's substring to be $m$. 

Again by applying the Hoeffding's inequality, with the random variable of interest Z to be the number of erroneous bits Bob obtains,
\begin{equation}
	\Pr [ |Z-p_{\rm err}m| \geq \alpha_2m] \leq 2e^{-2\alpha_{2}^{2}m} \leq \epsilon
\end{equation}
again by using $\alpha_2$ as defined in \eqref{param}. Hence the correctness of the protocol is guaranteed except for probability $2\epsilon$.
\end{proof}

\paragraph{Security against dishonest Alice}

We now proceed to prove security against Alice. Recall that a malicious Alice
can avoid bit flip errors and tamper with the bit string directly. We need to show that no matter how Alice tampers with the string,
Bob will detect her cheating with probability close to unity.\\

Previously~\cite{noisy:new} for $p_{\rm err} = 0$ , whenever Bob checks $X_\cI$ against $X^n$
and finds one faulty bit, he aborts the protocol directly. However, in a
realistic setup Bob accepts a number of roughly $p_{\rm err} m$ bits. By properties of the error-correcting code, we know from \cite{noisy:new} that for any attack of Alice she has to change at least $\frac{d}{2}$ such that the Bob will accept the syndrome to be 
consistent~\cite{noisy:new}. 
With this, we set a constraint on the code distance used such that whenever Alice attempts to cheat, Bob picks up enough faulty bits to detect the cheating except for some minimal probability.

\begin{lemma}[Security against Alice]\label{secA}
If Bob is honest, given that Alice and Bob use an error-correcting code with
minimum distance $d > \frac{2(p_{\rm err}+\alpha_2)(\frac{1}{2}-\alpha_1)n}{\frac{1}{2}-\alpha_3}$, the pair of protocols(Commit, Open) is $2\epsilon$-binding.
\end{lemma}

\begin{proof}
In our proof we assume that a malicious Alice can avoid all bit flip errors in
WSEE, so that the scenario reduces to WSE where $p_{\rm err}=0$. From the proof of correctness, we know that Bob obtains enough bits from WSEE except with probability $\epsilon$.\\

For any general attack Alice can attempt, to satisfy Bob's check on the syndrome she has to change at least $\frac{d}{2}$ bits in the original string $X_{n}$, where $d$ is the code distance~\cite{noisy:new}. 
Since Bob picks up each faulty bit with probability $\frac{1}{2}$, by defining W to be the number of faulty bits where Bob obtains in his substring, and applying Hoeffding's inequality,
\begin{equation}
	\Pr \left[ W \leq (\frac{1}{2}-\alpha_3)\frac{d}{2} \right] \leq e^{-\alpha_{3}^{2}d} = \epsilon.
\end{equation}
We see that Bob picks up at least $(\frac{1}{2}-\alpha_3)\frac{d}{2}$ flipped bits except with probability $\epsilon$, for $\alpha_3$ as defined in \eqref{param}.
Combining with the fact that Bob accepts at most $(p_{\rm err}+\alpha_2)m$ bits, we require
\begin{equation}
(\frac{1}{2}-\alpha_3)\frac{d}{2} > (p_{\rm err}+\alpha_2)m.
\end{equation}
The requirement for code distance is then given by 
\begin{equation}
d > \frac{2(p_{\rm err}+\alpha_2)(\frac{1}{2}-\alpha_1)n}{\frac{1}{2}-\alpha_3}.
\end{equation}
Hence generally when Alice and Bob use a code with minimum distance that satisfies the above requirement, whenever Alice attempts to cheat, the pair of protocols is proven to be $2\epsilon$-binding.
\end{proof}

\paragraph{Security against dishonest Bob}
Subsequently, we prove security against Bob. Recall that a cheating Bob can first make arbitrary measurements and store some classical information, and then keep some quantum information in this noisy-storage device. The overall state of Bob's system can then be described as a ccq-state $\rho_{X^n K \Theta \cF(\cQ)}$, where $K$ being Bob's classical information obtained from measurements, $\Theta$ being Alice's basis information, and $\cF(\cQ)$ being Bob's quantum information stored in an imperfect quantum memory. \\

To quantify the $\epsilon$-smooth min-entropy of Bob's information about Alice's string $X^n$, we proceed as in~\cite{noisy:new}: We first bound Bob's ignorance about $X^n$ based on his classical information $K$ alone. Second, we then relate this bound to his ignorance about $X^n$ given $K$ \emph{and} $\cF(\cQ)$. This yields security statements in terms of the classical capacity of $\cF$. Note that it is known that for very many channels better security bounds are possible in terms of the entanglement cost~\cite{entCost} and the quantum capacity~\cite{qcextract}. However, the classical capacity is still much better understood and offers explicit parameters for many interesting channels. In contrast, e.g. the quantum capacity of the depolarizing channel is not known.

To bound Bob's ignorance given $K$ alone, we invoke our results of \cite{finiten} as stated in the following theorem:
\begin{theorem}[Uncertainty relation~\cite{finiten}]\label{smooth}
If Alice is honest, the $\epsilon$-smooth min-entropy of Bob's information about $X^n$ is
\begin{equation}
\hmineps(X^n|\Theta^n K) \geq g(s)n +\frac{2\log\epsilon-1}{s},\\
\end{equation}
where
\begin{equation}
g(s) = \frac{-1}{s} \left[ \log(1+2^s)-(1+s)\right]
\end{equation}
for any $0<s\leq 1$.
\end{theorem}

Theorem \ref{smooth} gives us a bound on the min-entropy rate of Bob's classical information of $X^n$, whenever Alice is honest. We then use Lemma 2.2 from \cite{noisy:new} to bound the min-entropy rate of $X^n$ when he has the ccq-state $\rho_{X^n \cK\Theta\cF(\cQ)}$. 
\begin{lemma}[Min-entropy with quantum side information, \cite{noisy:new}]\label{qsideinfo}
Consider an arbitrary ccq-state $\rho_{XTQ}$, and let $\epsilon,\epsilon' \geq 0$ be arbitrary. Let $\cF : \cB(\cH_\cQ)\rightarrow B(\cH_\cQ)$ be an arbitrary CPTPM representing a quantum channel. Then
\begin{equation}
\mathrm{H}_{\mathrm{min}}^{\epsilon+\epsilon'}(X^n|T\cF(\cQ)) \geq -\log P_{\rm succ}^\cF \left(\lfloor \hmineps(X^n|T) - \log \frac{1}{\epsilon'} \rfloor\right),
\end{equation}
with
\begin{align}
P_{\rm succ}^\cF(Rn) := \max_{\{\rho_y\}_y,\{D_y\}_y} \frac{1}{2^{nR}} \sum_{y\in \01^{nR}} \tr(D_y \cF(\rho_y))\ ,
\end{align}
where the maximum is taken over all encodings $\{\rho_y\}$ and decoding POVMs $\{D_y\}_y$ of classical symbols $y$. 
\end{lemma}

Note that this bound is dependent on the properties of the quantum storage device $\cF$. For example, if the memory were to be of arbitrarily large size and noiseless, the success probability $P_{\rm succ}^\cF(Rn)$ is always 1, and the min-entropy of Bob's information about $X^n$ would be simply zero. However, the assumption of noisy and bounded storage comes in here to give a sufficiently high min-entropy which is crucial for the security proof. For simplicity in further proofs, we also introduce a simpler version, considering only bounded storage, which is a simple consequence of the chain rule and monotonicity of the min-entropy~\cite{renato:diss,serge:bounded}.

\begin{corollary}[Min-entropy for bounded quantum storage]\label{boundedqmstor}
Assuming the min-entropy of Bob's information of $X^n$ be $H_\infty^{\epsilon} (X^n|T)$, and Bob has a perfect quantum memory $\cQ$ that can store $S$ qubits. 
\begin{equation}
\hmineps(X^n|T\cQ) \geq \hmineps(X^n|T) - S .
\end{equation}
\end{corollary}

With the above Theorem \ref{smooth} and Lemma \ref{qsideinfo}, we can prove security against Bob in two steps. First, we show that if Alice is honest, the min-entropy rate of Bob's information about the string $X^n$ is lower bounded by $\lambda n$ for some $\lambda$. Then, by using privacy amplification, we show that the cq-state of $C^\mathit{l}$ and Bob's information is $2\epsilon$-close to a product state, with $C^\mathit{l}$ having uniform distribution over $\lbrace0,1\rbrace^\mathit{l}$. As we are only interested in commiting a single bit in this experiment, we will restrict our statement to the case of $l=1$. A more general statement can be derived analogously.

\begin{lemma}[Security against Bob]\label{secB}
For a fixed parameter $\epsilon$, define $\lambda$ to be the $\frac{\epsilon}{2}$-smooth min-entropy rate of Bob's information about $X^n$. If Alice is honest, and if the code rate satisfies
\begin{equation}\label{condonR}
R\geq 1-\lambda+\frac{2\log\frac{1}{\epsilon}}{n},
\end{equation} 
then the pair of protocols (Commit, Open) are $2\epsilon$-hiding for $l=1$, i.e. the commitment of a single bit.
\end{lemma}
\begin{proof}
After executing WSEE, by using Theorem \ref{smooth} and Theorem~\ref{boundedqmstor}, Bob's $\frac{\epsilon}{2}$-smooth min-entropy about $X^n$ can be evaluated. Note that by sending the syndrome he obtains additional information about $X^n$, and this is accounted for by the chain rule and 
monotonicity property of min-entropy~\cite{renato:diss},
\begin{equation}
H_{\mathrm{min}}^{\epsilon_o} (X^n | {\rm B'}, \Syn (X^n)) \geq H_{\mathrm{min}}^{\epsilon_o} (X^n| {\rm B'}) -\mathit{L}.
\end{equation}
where $\mathit{L}$ is the length of the syndrome.
Recall that the length of syndrome is $L=n-k=(1-R)n$, where $R=\frac{k}{n}$ is the code rate. Hence, we have
\begin{equation}
H_{\mathrm{min}}^{\epsilon/2} (X^n | {\rm B'}, \Syn (X^n)) \geq (\lambda - 1 + R)n
\end{equation}
which denotes the min-entropy rate of Bob's total information about $X^n$ at the end of the commit phase.\\

Next, we show that by privacy amplification Bob does not gain knowledge about the committed information $C^\mathit{l}$. Denoting the committed string as $ C^\mathit{l}=\Ext(X^n,R)\in{\{0,1\}}^\mathit{l}$ having length $\mathit{l}$ we have from~\cite{renato:diss} that
\begin{equation}\label{priamp}
\rho_{C^\mathit{l},{\rm B'}\Syn(X^n)} \approx_{\epsilon '} \tau_{{\{0,1\}}^\mathit{l}} \otimes \rho_{{\rm B'},\Syn(X^n)}
\end{equation}
where 
\begin{equation}\label{priamperror}
\epsilon ' = 2\epsilon_o + 2^{-\frac{1}{2}[\SH_{\mathrm{min}}^{\epsilon_o} (X^n|{\rm B'},\Syn(X^n))-\mathit{l}]-1}
\end{equation}
and $\tau_\cA$ is the uniform distribution over the entire set $\cA$. Setting $\epsilon_o=\frac{\epsilon}{2}$,
\begin{equation}\label{priamp2}
\epsilon ' = \epsilon + \frac{1}{2}\cdot 2^{-\frac{1}{2}[\SH_{\mathrm{min}}^{\epsilon_o} (X^n|{\rm B'},\Syn(X^n))-\mathit{l}]}
\end{equation}

Setting the second term in \eqref{priamp2} to be $\epsilon$, and setting $\mathit{l}=1$, this implies
\begin{equation}
\lambda -1+R -\frac{1}{n}> -\frac{2\log\epsilon+2}{n},
\end{equation}

Rearranging gives \eqref{condonR}. In the large n limit, we require $R>1-\lambda$. 
\end{proof}

With this, we end the security proof against Bob.\\

In summary, we have derived conditions on the relative minimum distance $\delta=\frac{d}{n}$ and the code rate R for where the protocol is secure. By combining Lemma \ref{correctness}, Lemma \ref{secA} and Lemma \ref{secB}, we summarize these results into the following theorem:
\begin{theorem}[Conditions for successful execution of the BC protocol]\label{summary}
Let $n \in \mathbb{N}$, $\epsilon > 0$ and $\lambda > 0$. If the error correcting code used satisfies the following requirements:
\begin{list}{*}{}
\item[Relative minimum distance:] $\delta > \frac{2(p_{\rm err}+\alpha_2)(\frac{1}{2}-\alpha_1)}{\frac{1}{2}-\alpha_3}$.
\item[Code rate:] $R > 1-\lambda + \frac{2\log\frac{1}{\epsilon}}{n}$,
\end{list}
then Protocol 2 is $2\epsilon$-correct, $2\epsilon$-binding and $2\epsilon$-hiding.
\end{theorem}

The final part of the theoretical analysis is to discuss the feasibility of finding an error-correcting code that satisfies the requirements stated in Theorem \ref{summary}. Clearly, there exists a trade-off between parameters of code rate R and relative minimum distance $\delta$. We make use of Theorem \ref{randomcode} to argue that once the parity check matrix of a code with rate R is randomly generated, its distance is lower bounded except for an extremely small probability. Subsequently, by using Theorem \ref{summary} and Theorem \ref{randomcode}, we evaluate an optimal parameter $n=2.5\times 10^5$ in the \suppmtd~\ref{ourexp}, which is used in our experiment, and show that for such a block length, the bit commitment protocol is secure except for an error $3\cdot 10^{-4}$.

Also, by combining Theorem \ref{summary} with Theorem \ref{randomcode}, we provide a cleaner expression for the bound on minimum distance accompanied by a lower bound on the block length such that security can be achieved:

\begin{theorem}[Secure bit commitment] \label{security}
Let $\epsilon > 0$, $\lambda \geq 0.3$,  $\beta\in\left(0,0.01\right]$, $\delta \in [0.05,0.11]$, and
\begin{equation}\label{boundm3}
n\geq\frac{1}{\beta^2}\log\frac{2}{\epsilon}\ .
\end{equation}
Also, denote $h(x)=-x\log x-(1-x)\log(1-x)$ as the binary entropy function. 

If the bit flip error rate satisfies
\begin{equation}\label{condperr}
p_{\rm err} < (1-4\sqrt{5}\beta)\cdot\frac{\delta}{2} -\frac{\beta}{\sqrt{1-2\beta}} ,
\end{equation}
and the smooth min-entropy rate satisfies 
\begin{equation}\label{condlambda}
\lambda > h(\delta)+3\beta^2,
\end{equation}
then the bit commitment protocol is $3\epsilon$-secure by using a randomly generated error-correcting code.
\end{theorem}
\begin{proof}
Fix $\epsilon$ and assume that
\begin{equation}
 \sqrt{\frac{\ln\frac{2}{\epsilon}}{n}}\leq  \sqrt{\frac{\log\frac{2}{\epsilon}}{n}}\leq \beta.
\end{equation} 
This is achieved for any $\beta$, provided \eqref{boundm3} is true. Hence from \eqref{param}, we have
\begin{equation}\label{alpha1}
0\leq \alpha_1 \leq \beta.
\end{equation}
Plugging this into $m=(\half-\alpha_1)n$ gives
\begin{equation}\label{alpha2}
\alpha_2 = \sqrt{\frac{\ln\frac{2}{\epsilon}}{2m}} \leq \frac{\beta}{\sqrt{1-2\beta}}.
\end{equation}
Meanwhile, assume that $\delta\geq 0.05$. This leads to the condition that $\lambda\geq 0.3$ which is generally achieved.
\begin{equation}\label{alpha3}
\alpha_3 \leq \sqrt{20}\beta.
\end{equation}
Plugging \eqref{alpha1}, \eqref{alpha2}, and \eqref{alpha3} into the condition on minimum distance given in Theorem \ref{summary}, we obtain
\begin{equation}\label{condondelta}
\delta > 2\cdot\frac{p_{\rm err}+\frac{\beta}{\sqrt{1-2\beta}}}{1-4\sqrt{5}\beta} \geq \frac{2(p_{\rm err}+\alpha_2)(\frac{1}{2}-\alpha_1)}{\frac{1}{2}-\alpha_3}.
\end{equation}
Rearranging, we obtain \eqref{condperr}. 

To derive \eqref{condlambda}, we first use Theorem \ref{randomcode}. By setting the additional error from the generation of random codes to be smaller than $\epsilon$, we obtain
\begin{equation}
R<1-h(\delta)-\frac{\log\frac{1}{\epsilon}}{n}.
\end{equation}
Combining this with the condition on code rate R given in Theorem \ref{summary}, we have
\begin{eqnarray}
1-h(\delta)&>& 1-\lambda+2\frac{\log\frac{1}{\epsilon}}{n}\nonumber\\
\lambda &>& h(\delta) + 3\frac{\log\frac{1}{\epsilon}}{n},
\end{eqnarray}
which is satisfied if \eqref{condlambda} is true.
The total execution error of Protocol 2 becomes $3\epsilon$, where $2\epsilon$ comes from the execution error in Theorem \ref{summary}, and additional $\epsilon$ accounts for the probability that a randomly generated error-correcting code does not fulfill the requirement on relative minimum distance $\delta$.
\end{proof}

\eqref{boundm3} provides us with a lower bound on $n$ for a secure implementation. Note that this lower bound is non-tight, due to the approximations made while deriving bounds for \eqref{alpha1}, \eqref{alpha2}, and \eqref{alpha3}. Also, it is stressed that Theorem \ref{randomcode} gives a general proof for a randomly generated error-correcting code. 
This approach is taken because systematic ways of constructing such binary linear codes are not known, and the task of evaluating the minimum distance of a given code is NP-hard. However, as outlined in the \supdis~\ref{section:codes}
it is well known that the probability of generating a code with undesirable properties is minimal and added as a source of error in the protocol. 
It is also worth noting that a random code allows for easy execution of the protocol, since the only computation involved for honest parties is the calculation of the syndrome. That is, Alice and Bob never need to decode.

\subsection{Range of experimental parameters for implementation}\label{section:exp}

In this section, we provide full statements about the security of commitments, by combining the analysis accounting for both erasures (WSEE) and errors (BC). We work towards a simplified expression for the rate of commitment, i.e. to commit one bit securely, what is the required number of signals to send. We present region plots showing where security holds for the protocol. 

For security, we first note that the analysis of bit flip errors requires a minimum guaranteed amount of min-entropy $\hat{\lambda}$ for the commitment to be secure. This is seen in Lemma \ref{secB} where the lower bound for code rate R can be translated into a lower bound for the min-entropy rate. The main condition for feasibility of bit commitment is then given by
\begin{equation}
\lambda \geq \hat{\lambda}.
\end{equation}

The following theorem shows $\hat{\lambda}$ by considering a randomly generated binary linear code.
\begin{theorem}
By fixing an error parameter $\epsilon$ which indicates the error for the
generation of a random code. If and given the parameter $\delta$, which is the
relative minimum distance of the code required, as determined by $p_{\rm err}$ by using \eqref{condondelta},
\begin{equation}
\hat\lambda = h(\delta)+\frac{3\log\frac{1}{\epsilon}}{n}.
\end{equation}
\end{theorem}

On the other hand, we need to evaluate a bound on the min-entropy rate
$\lambda$ created according to the parameters $p_{\rm sent}^1$, $p_{\rm B, no
  click}^{\rm h}$ and $p_{\rm B, no click}^{\rm d}$. We begin by defining two relevant fractional quantities:
\begin{align}
m_{\rm left}^1&=p^1_{\rm sent} - p_{\rm B, no click}^{\rm h} + p_{\rm B, no click}^{\rm d} - 3\zeta\\
m_{\rm frac}&=1- p_{\rm B, no click}^{\rm h}-\zeta
\end{align}
where $\zeta = \sqrt{\frac{\ln\frac{2}{\epsilon}}{2M}}$. 
The smooth min-entropy rate can then be evaluated, by invoking the bounded, or in general noisy storage assumption. We refer to below for some examples.

\subsection*{Example: bounded storage}
We state a theorem describing the number of signals M needed to send for a secure commitment, given the relevant experimental parameters and assuming the case of bounded storage. It is worth stressing again that $n$ denotes the \textbf{block length used in the commitment}, while $M$ denotes the \textbf{number of signals sent from Alice to Bob}. These quantities are related by the expressions given in \eqref{param}.

\begin{lemma}\label{d2}
Let a dishonest Bob's storage size be bounded by S. For fixed parameters S and $\epsilon$, and given the experimental probabilities listed in Table I, and for some $\beta,\gamma \in (0,0.01]$, let
\begin{eqnarray}
m_2 &=& p_{\rm sent}^1-p_{\rm B, no click}^{\rm h}+p_{\rm B, no click}^{\rm d} -3\gamma\nonumber\\
m_3 &=& 1-p_{\rm B, no click}^{\rm h}\nonumber\\
L'&=&\max_{ s \in (0,1] } ~ \frac{-1}{s} \left[\log(1+2^s)-1-s\right] -\frac{3\epsilon}{s}\nonumber\\
\hat{\lambda} &=& h(\delta) + 3\beta^2\nonumber\\
\delta &=& 2\cdot\frac{p_{\rm err}+\frac{\beta}{\sqrt{1-2\beta}}}{1-4\sqrt{5}\beta}\nonumber\\
M_1 &=& \frac{1}{2\gamma^2}\log\frac{2}{\epsilon}\nonumber\\
M_2 &=& \frac{\log\frac{1}{\epsilon}}{\epsilon\cdot m_2}\nonumber\\
M_3 &=& \frac{\log\frac{2}{\epsilon}}{(m_3-\gamma)\beta^2}\nonumber\\
M_4 &=& \frac{S}{m_2 L'-m_3\hat{\lambda}}
\end{eqnarray}
For security to hold at all, the following is required:	
\begin{equation}\label{condition}
m_2 L' -m_3 \hat{\lambda} > 0.
\end{equation}
If \eqref{condition} is true, then bit commitment can be implemented $3\epsilon$-securely by using a randomly constructed error-correcting code, whenever
\begin{equation}\label{boundingm}
M > \max ~ \{M_1,M_2,M_3,M_4\}.
\end{equation}

\begin{proof} By the analysis of \cite{Curty10}, the min-entropy rate has the form
\begin{equation}
\lambda = \frac{m_{\rm left}^1\cdot L - \frac{S}{M}}{m_{\rm frac}},
\end{equation}
where 
\begin{equation}
L = \max_{ s \in (0,1] } ~ \frac{-1}{s} \left[\log(1+2^s)-1-s\right] -\frac{3\log\frac{1}{\epsilon}}{m_{\rm left}^1\cdot M}.
\end{equation}
Note that $\zeta$ is dependent on $M$, and its value decreases while $M$ increases. Setting $\zeta \leq \gamma$ which is a constant, provides an lower bound for $M$, which gives the value for $M_1$ depending on the chosen $\gamma$.\\
Similarly, setting $\frac{\log\frac{2}{\epsilon}}{m_{\rm left}^1\cdot M}\leq\epsilon$ provides $M_2$ and $L\geq L'$. \\
$M_3$ comes from the condition given for $n$ at \eqref{boundm3}, while $M=\frac{n}{m_{\rm frac}} \geq \frac{n}{m_3-\gamma}$.
Lastly, 
\begin{equation}
\lambda\geq \frac{m_2\cdot L' - \frac{S}{M}}{m_3} \geq \hat{\lambda}
\end{equation}
provides the main condition for security to hold at all (S=0), while rearranging gives the value for $M_4$.
\end{proof}
\end{lemma}

\subsection*{Example: noisy storage}
Lemma \ref{d2} gives the case for bounded storage model, with no quantum noise assumed for cheating Bob's storage device. For a more general noisy storage assumption, where the quantum channel satisfies a strong converse relation as in \cite{noisy:new}, the $\frac{\epsilon}{2}$-smooth min-entropy can be evaluated by
\begin{align}\label{depolarizingstorage}
& H_{\mathrm{min}}^{\epsilon'+\epsilon''}(X^n|\Theta^n K \cF(\cQ)) \nonumber\\
&\geq -\log P_{\rm succ}^{\cN^{\otimes S}} \left[H_{\mathrm{min}}^{\epsilon'}(X^n|\Theta^n K)-\log\frac{1}{\epsilon''}\right]\nonumber\\
&= S\cdot \gamma^{\cN} \left(\frac{H_{\mathrm{min}}^{\epsilon'}(X^n|\Theta^n K)-\log\frac{1}{\epsilon''}}{S}\right),
\end{align}
where $\gamma^{\cN}$ is the strong converse parameter of the quantum channel. For a fixed error parameter $\epsilon$, $\epsilon'$ and $\epsilon''$ should be chosen such that $\epsilon'+\epsilon'' = \frac{\epsilon}{2}$. Compared to the analysis in \cite{Curty10}, where the storage size is determined by introducing a \textit{storage rate} quantity $\nu$, in this analysis we work with the quantity $S$, which is the maximum number of qubits Bob is able to store. To use this quantity in the analysis, we invoke the bounded storage assumption that Bob cannot store more than $S$ qubits, then calculate the conditions for security. For example, in the setting of depolarizing noise for a two-dimensional quantum channel, the strong converse parameter is given by the following expression \cite{noisy:new}:
\begin{eqnarray}
\gamma^{\cN} (\hat{R}) &=& \max_{\alpha\geq 1} \frac{\alpha-1}{\alpha} (\hat{R}- C) \nonumber\\
C &=& 1 - \frac{1}{1-\alpha}\log\left[ p^\alpha+(1-p)^\alpha\right] \nonumber\\
p &=& \frac{1+r}{2}.
\end{eqnarray}

\subsection*{Our experiment}\label{ourexp}
We state again the bounds of experimental parameters as derived:
\begin{eqnarray}
p_{\rm sent}^1&=&1-p_{\rm sent}^0-p_{\rm sent}^{n>1}>0.125\nonumber\\
p_{\rm sent}^0+p_{\rm sent}^1&=&1-p_{\rm sent}^{n>1}>0.99947\nonumber\\
p_{\rm B, no click}^{\rm h}&=&0.909,\nonumber\\
p_{\rm err} &=& 0.0412.
\end{eqnarray}
Note that $p_{\rm B, no click}^{\rm h}$ and $p_{\rm err}$ are values obtained after the symmetrization procedures on both Alice's and Bob's side.

Bounds on $n$ (and $M$) derived based on Theorem \ref{security} and Lemma \ref{d2} are non-tight. Here, we use an optimal block length such that the classical information post-processing is minimal. We summarize the calculations in the following steps:

\begin{list}{\arabic{qcounter}:~}{\usecounter{qcounter}}
\item {\bf Fix $\epsilon$ and $n$:} Firstly, we set $\epsilon=0.99\times 10^{-5}$, and $n=2.5\times 10^5$. By doing so, all relevant parameters in \eqref{param} can be evaluated, except for $\alpha_3$ which will depend on the error-correcting code. $M$ is known by its relation with $n$ as stated in \eqref{param}, where this is justified by a detailed explanation offered right after Protocol 1.

\item {\bf Evaluate relative minimum distance $\delta$:} By performing a numerical optimization that satisfies the condition on relative distance, as stated in Theorem \ref{summary}, we obtain $\delta > 0.998201$. 

\item {\bf Set $\epsilon_{code}$:} To obtain a code that with $\delta$ that satisfies the condition as evaluated in step 2, we use Theorem \ref{randomcode}. First, we need to set an $\epsilon_{code}=2\times 10^{-7}$, which bounds the probabilistic error for generating a bad random code. By doing so, we pose an upper bound upon the code rate $R$. Using $R=0.531$ satisfies this condition. By using Theorem \ref{summary}, the protocol is secure provided that the $\frac{\epsilon}{2}$-smooth min entropy rate $\lambda > 0.469133$.

\item {\bf Evaluate storage assumption:} By using $p_{\rm B, no click}^{\rm h}$, $p_{\rm B, no click}^{\rm h}$, $p_{\rm sent}^1$, $M$ and $\epsilon$, evaluate and optimize $\lambda$ for different storage noise and storage sizes $S$, such that $\lambda$ satisfies the condition in Step 3.

\item {\bf Evaluate total execution error:} The total execution error is then evaluated as $\epsilon_{total}=2\epsilon+\epsilon_{code} = 2\times 10^{-5}$.
\end{list}

For the bounded storage assumption, the commitment is secure whenever dishonest Bob's storage size is bounded by $S_{bounded}=928$ qubits.
For the noisy storage assumption, we can use \eqref{depolarizingstorage} and maximize over all choices of $\epsilon'$ and $\epsilon''$. For a depolarizing noise of noise parameter $r=0.9$, the commitment is secure whenever Bob's storage size is bounded by $S_{noisy} = 972$ qubits.

\end{document}